\documentclass[10pt,twocolumn,twoside]{IEEEtran}

\usepackage{amsmath}
\usepackage{subfigure}
\usepackage{amssymb}
\usepackage{stmaryrd}
\usepackage{verbatim}
\usepackage{graphicx}
\usepackage{graphics}
\usepackage{theorem}
\usepackage{tabularx}
\usepackage{array}
\usepackage{bbm}
\usepackage[noadjust]{cite}
\usepackage{algorithmic}
\usepackage{algorithm}

\usepackage[colorlinks=true, pdfstartview=FitV, linkcolor=blue, citecolor=blue, urlcolor=blue]{hyperref}

\def\Gr#1{\mathbf{#1}}

\theoremstyle{break}

\newtheorem{Thm}{Theorem}[section]
\newtheorem{Lemme}{Lemma}[section]
\newtheorem{Coro}{Corollary}[section]
\newtheorem{Prop}{Proposition}[section]
\newtheorem{rem}{Remark}[section]

\newcommand\Mc{\widehat{\Gr{M}}}

\newcommand\be{\begin{pmatrix}}
\newcommand\en{\end{pmatrix}}
\newcommand\bep{\left(}
\newcommand\enp{\right)}

\newcommand\disp{\displaystyle}

\def\R{\mathbb{R}}

\def\D{\mathcal{D}}

\def\NN{\mathcal{N}}
\def\A{\mathcal{A}}

\def\L{\mathcal{L}}

\def\EOP{\ \hfill \rule{0.5em}{0.5em} }

\def\Mm{M_p(\R)}

\def\Mt{\Gr{M}^{T}}

\def\Db{\overline{\D}}

\newcommand{\tr}{\,\text{Tr}}

\renewcommand{\top}{T}

\begin{document}
\title{Parameter Estimation For Multivariate Generalized Gaussian Distributions}


\author{Fr\'ed\'eric Pascal, Lionel Bombrun, Jean-Yves Tourneret and Yannick Berthoumieu.
\IEEEcompsocitemizethanks{
\IEEEcompsocthanksitem F. Pascal is with Sup\'elec/SONDRA, 91192 Gif-sur-Yvette Cedex, France (e-mail: frederic.pascal@supelec.fr)
\IEEEcompsocthanksitem L. Bombrun and Y. Berthoumieu are with Universit\'e de Bordeaux, IPB, ENSEIRB-Matmeca, Laboratoire IMS, France (e-mail: lionel.bombrun@ims-bordeaux.fr; yannick.berthoumieu@ims-bordeaux.fr)
\IEEEcompsocthanksitem J.-Y. Tourneret is with Universit\'e de Toulouse, IRIT/INP-ENSEEIHT, (e-mail:jean-yves.tourneret@enseeiht.fr)}
}

\markboth{Submitted to IEEE Trans. on Signal Processing}{Pascal
\MakeLowercase{\textit{et al.}}: Parameter Estimation For Multivariate Generalized Gaussian Distributions}

\maketitle

\begin{abstract}
Due to its heavy-tailed and fully parametric form, the multivariate generalized Gaussian distribution (MGGD) has been receiving much attention for modeling extreme events in signal and image processing applications. Considering the estimation issue of the MGGD parameters, the main contribution of this paper is to prove that the maximum likelihood estimator (MLE) of the scatter matrix exists and is unique up to a scalar factor, for a given shape parameter $\beta\in(0,1)$. Moreover, an estimation algorithm based on a Newton-Raphson recursion is proposed for computing the MLE of MGGD parameters. Various experiments conducted on synthetic and real data are presented to illustrate the theoretical derivations in terms of number of iterations and number of samples for different values of the shape parameter. The main conclusion of this work is that the parameters of MGGDs can be estimated using the maximum likelihood principle with good performance.

\end{abstract}


\IEEEpeerreviewmaketitle


\section{Introduction}
\label{sec:intro}
\PARstart{U}{nivariate} and multivariate generalized Gaussian distributions (GGDs) have received much attention in the literature. Historically, this family of distributions has been introduced  in \cite{subbotin1923law}. Some properties of these distributions have been reported in several papers such as~\cite{Gomez1998,Gomez2002,Song2006}. These properties include various stochastic representations, simulation methods and probabilistic characteristics.  GGDs belong to the family of elliptical distributions (EDs)~\cite{Fang-90a, Fang-90b}, originally introduced by Kelker in \cite{kelker1970distribution} and studied in~\cite{Rangaswamy1993,Rangaswamy1995}. For $\beta \in (0,1]$, Multivariate GGDs (MGGDs) are a subset of the spherically invariant random vector (SIRV) distributions. For $\beta>1$, MGGDs are no longer SIRV distributions as illustrated  in Fig.~\ref{fig:SIRV_MGDD} (for more details, see \cite{gomez2008multivariate}).



MGGDs have been used intensively in the image processing community. Indeed, including Gaussian and Laplacian distributions as special cases, MGGDs are potentially interesting for modeling the statistical properties of various images or features extracted from these images. In particular, the distribution of wavelet or curvelet coefficients has been shown to be modeled accurately by GGDs \cite{Mallat1989,Chang2000a,Chang2000b,Boubchir2005}. This property has been exploited for many image processing applications including image denoising \cite{Moulin1999,Sendur2002,Cho2005,Cho2009}, context-based image retrieval \cite{Do2002,Verdoolaege2008}, image thresholding \cite{Bazi2007} or texture classification in industrial problems \cite{Scharcanski2007}. Other applications involving GGDs include radar \cite{Desai2003}, video coding and denoising \cite{Coban1996,Bicego2008,Yang2009} or biomedical signal processing \cite{Bicego2008,Elguebaly2010,LeCam2009}. Finally, it is interesting to note that complex GGDs have been recently studied in \cite{Novey2001,Novey2009}  and that multivariate regression models with generalized Gaussian errors have been considered in \cite{liu2008multivariate}.

Considering the important attention devoted to GGDs, estimating the parameters of these distributions is clearly an interesting issue. Classical estimation methods that have been investigated for univariate GGDs include the maximum likelihood (ML) method \cite{agro1995maximum} and the method of moments \cite{Varanasi1989}. In the multivariate context, MGGD parameters can be estimated by a least-squares method as in~\cite{Cho2005} or by minimizing a $\chi^2$ distance between the histogram of the observed data and the theoretical probabilities associated with the MGGD~\cite{Khel-04}. Estimators based on the method of moments and on the ML method have also been proposed in~\cite{Verd-11a,Verd-11b,zhang2013multivariate}.

Several works have analyzed covariance matrix estimators defined under different modeling assumptions. On the one hand, fixed point (FP) algorithms have been derived and analyzed in~\cite{Pasc-08, chitour2008exact} for SIRVs. On the other hand, in the context of robust estimation, the properties of M-estimators have been studied by Maronna in~\cite{Maro-76}. Unfortunately, Maronna's conditions are not fully satisfied for MGGDs (see remark~\ref{rem:3}). This paper shows that despite the non-applicability of Maronna's results, the MLE of MGGD parameters exists, is unique and can be computed by an FP algorithm. Although the methodology adopted in this paper has some similarities with the one proposed in \cite{Pasc-08, chitour2008exact}, there are also important differences which require a specific analysis (see for instance remark~\ref{rem:4}). More precisely, the FP equation of \cite{Pasc-08} corresponds to an approximate MLE for SIRVs while in \cite{chitour2008exact} the FP equation results from a different problem (see Eq.~(14) in \cite{chitour2008exact} compared to Eq.~\eqref{functionF} of this paper). The contributions of this paper are to establish some properties related to the FP equation of the ML estimator for MGGDs. More precisely, we show that for a given shape parameter $\beta$ belonging to $(0,1)$, the MLE of the scatter matrix $\mathbf{M}$ exists and is unique up to a scalar factor\footnote{From the submission of this paper, another approach based on geodesic convexity was proposed in (include reference paper Wiesel).}. An iterative algorithm based on a Newton-Raphson procedure is then proposed to compute the MLE of $\mathbf{M}$.\\

The paper is organized as follows. Section~\ref{sec:formulation} defines the MGGDs considered in this study and derives the MLEs of their parameters. Section~\ref{sec:main} presents the main theoretical results of this paper while a proof outline is given in Section~\ref{section4}. For presentation clarity, full demonstrations are provided in the appendices. Section~\ref{sec:simultations} is devoted to simulation results conducted on synthetic and real data. The convergence speed of the proposed estimation algorithm as well as the bias and consistency of the scatter matrix MLE are first investigated using synthetic data. Experimentations performed on real images extracted from the VisTex database are then presented. Conclusions and future works are finally reported in Section~\ref{sec:conclusion}.



\section{Problem formulation} \label{sec:formulation}

\subsection{Definitions}
The probability density function of an MGGD in $\mathbb{R}^p$ is defined by~\cite{Kotz-68}
\begin{align}
\label{eq:pdf_MGGD_m}
p(\mathbf{x}\vert \mathbf{M}, m, \beta) & = \dfrac{1}{\vert\mathbf{M}\vert^\frac{1}{2}} h_{m,\beta}\left( \mathbf{x}^T \mathbf{M}^{-1} \mathbf{x} \right)
\end{align}
for any $\mathbf{x} \in  \mathbb{R}^p$, where $\mathbf{M}$ is a $p\times p$ symmetric real scatter matrix, $\mathbf{x}^T$ is the transpose of the vector $\mathbf{x}$, and $h_{m,\beta}\left( \cdot \right)$ is a so-called density generator defined by
\begin{equation} \label{eq:function_h}
h_{m,\beta}\left( y \right) = \dfrac{\beta \Gamma\left(\frac{p}{2}\right)}{\pi^{\frac{p}{2}} \Gamma\left( \frac{p}{2\beta} \right) 2^{\frac{p}{2\beta}}} \dfrac{1}{m^\frac{p}{2}}
\exp\left( -\dfrac{y^{\beta}    }{2m^\beta}  \right)
\end{equation}
for any $y \in \mathbb{R}^+$, where $m$ and $\beta$ are the MGGD scale and shape parameters. The matrix $\mathbf{M}$ will be normalized in this paper according to $\tr\left( \mathbf{M} \right)=p$, where $\tr(\mathbf{M} )$ is the trace of the matrix $\mathbf{M} $. It is interesting to note that letting $\beta=1$ corresponds to the multivariate Gaussian distribution. Moreover, when $\beta$ tends toward infinity, the MGGD is known to converge in distribution to a multivariate uniform distribution (see~\eqref{eq:representation_stochastique}).

\subsection{Stochastic representation}
Let $\mathbf{x}$ be a random vector of $\mathbb{R}^p$ distributed according to an MGGD with scatter matrix $\mathbf{\Sigma}=m\mathbf{M}$ and shape parameter $\beta$. G\'{o}mez \textit{et al.} have shown that $\mathbf{x}$ admits the following stochastic representation~\cite{Gomez1998}
\begin{align}
\label{eq:representation_stochastique}
\mathbf{x} \overset{d}{=} \tau ~ \mathbf{\Sigma}^{\frac{1}{2}} ~ \mathbf{u}
\end{align}
where $\overset{d}{=}$ means equality in distribution, $\mathbf{u}$ is a random vector uniformly distributed on the unit sphere of $\mathbb{R}^p$, and $\tau$ is a scalar positive random variable such that
\begin{align}
\tau^{2\beta} \sim \Gamma \left( \frac{p}{2\beta}, 2 \right)
\end{align}
where $\Gamma (a,b)$ is the univariate gamma distribution with parameters $a$ and $b$ (see \cite{John-94} for definition).

\subsection{MGGD parameter estimation for known $\beta$}
Let $(\mathbf{x}_1, \ldots, \mathbf{x}_N)$ be $N$ independent and identically distributed (i.i.d.) random vectors distributed according to an MGGD with parameters $\mathbf{M}, m$ and $\beta$. This section studies estimators of $\mathbf{M}$ and $m$ based on $(\mathbf{x}_1, \ldots, \mathbf{x}_N)$ for a known value of $\beta \in (0,1)$\footnote{We note here that most values of $\beta$ encountered in practical applications belong to the interval $(0,1)$. For instance, $\beta = 0.8$ is suggested in \cite{hernandez2000dct} as a good choice for most images.}. The MGGD is a particular case of elliptical distribution that has received much attention in the literature. Following the results of \cite{Gini-02} for real elliptical distributions, by differentiating the log-likelihood of vectors $(\mathbf{x}_1, \ldots, \mathbf{x}_N)$ with respect to $\mathbf{M}$, the MLE of the matrix $\mathbf{M}$ satisfies the following FP equation
\begin{align}
\mathbf{M} = \dfrac{2}{N} \sum\limits_{i=1}^N \dfrac{-g_{m,\beta}(\mathbf{x}_i^T \mathbf{M}^{-1} \mathbf{x}_i)}{h_{m,\beta}(\mathbf{x}_i^T \mathbf{M}^{-1} \mathbf{x}_i)} \mathbf{x}_i \mathbf{x}_i^T
\end{align}
where $g_{m,\beta}(y) = \partial h_{m,\beta}(y) / \partial y$. In the particular case of an MGGD with known parameters $m$ and $\beta$, straightforward computations lead to
\begin{align}
\label{eq:point_fixe_1}
\mathbf{M} = \dfrac{\beta}{N m^{\beta}} \sum\limits_{i=1}^N \dfrac{\mathbf{x}_i \mathbf{x}_i^T}{ \left( \mathbf{x}_i^T \mathbf{M}^{-1} \mathbf{x}_i \right)^{1-\beta}}.
\end{align}
When the parameter $m$ is unknown, the MLEs of $\mathbf{M}$ and $m$ are obtained by differentiating the log-likelihood of $(\mathbf{x}_1, \ldots, \mathbf{x}_N)$ with respect to $\mathbf{M}$ and $m$ yielding
\begin{align}
\mathbf{M} &= \dfrac{\beta}{N m^{\beta}} \sum\limits_{i=1}^N \dfrac{\mathbf{x}_i \mathbf{x}_i^T}{ \left( \mathbf{x}_i^T \mathbf{M}^{-1} \mathbf{x}_i \right)^{1-\beta}}, \label{MvM}\\
m &= \left[ \dfrac{\beta}{pN} \sum\limits_{i=1}^N \left( \mathbf{x}_i^T \mathbf{M}^{-1} \mathbf{x}_i \right)^{\beta} \right]^{\frac{1}{\beta}} \label{Mvm}.
\end{align}
After replacing $m$ in \eqref{eq:point_fixe_1} by its expression \eqref{Mvm}, the following result can be obtained
\begin{align}
\label{eq:point_fixe_2}
\mathbf{M} = \dfrac{1}{N} \sum\limits_{i=1}^N \dfrac{Np}{y_i + y_i^{1-\beta} \disp \sum_{j\neq i} y_j^{\beta} } \,\mathbf{x}_i \mathbf{x}_i^T\,.
\end{align}
As mentioned before and confirmed by \eqref{eq:point_fixe_2}, $\mathbf{M}$ can be estimated independently from the scale parameter $m$. Moreover, the following remarks can be made about \eqref{eq:point_fixe_2}.

\begin{rem}
\label{rem:1}
When $\beta=1$, Eq.~\eqref{eq:point_fixe_2} is close to the sample covariance matrix (SCM) estimator (the only difference between the SCM estimator and \eqref{eq:point_fixe_2} is due to the estimation of the scale parameter that equals $1$ for the multivariate Gaussian distribution). For $\beta=0$, \eqref{eq:point_fixe_2} reduces to the FP covariance matrix estimator that has received much attention in \cite{Gini-02,Wiesel2012,Pasc-08a}.
\end{rem}

\begin{rem}
\label{rem:2}
Equation~\eqref{eq:point_fixe_2} remains unchanged if $\mathbf{M}$ is replaced by $\alpha\,\mathbf{M}$ where $\alpha$ is any non-zero real factor. Thus, the solutions of \eqref{eq:point_fixe_2} (when there exist) can be determined up to a scale factor $\alpha$. The normalization $\tr \left( \mathbf{M} \right)=p$ will be adopted in this paper and will be justified in the simulation section.
\end{rem}

\begin{rem}
\label{rem:3}
Let us consider the function $f_i$ defined by
\begin{align}
\label{eq:u_Maronna}
f_i(y) = \cfrac{Np}{y + c_i y^{1-\beta}}, \; \forall y \in \mathbb R^+
\end{align}
where $c_i$ is a positive constant independent of $y$ (the index $i$ is used here to stress the fact that $c_i = \disp \sum_{j\neq i} y_j^{\beta}$ changes with $i$ but does not depend on $y_i$). Equation~\eqref{eq:point_fixe_2} can be rewritten as
\begin{align}
\label{eq:point_fixe_2_bis_Mar}
\mathbf{M} = \dfrac{1}{N} \sum\limits_{i=1}^N f_i \left(y_i \right) \,\mathbf{x}_i \mathbf{x}_i^T.
\end{align}
Roughly speaking\footnote{Actually, Maronna's function depends only on the $i^{th}$ sample and not on all the samples as it is the case here!}, $f_i$ satisfies Maronna's conditions (recapped below, see~\cite[p. 53]{Maro-76} for more details) for any $\beta \in (0,1)$ except the continuity at $y=0$.\\ \textit{Maronna's conditions for a function $f: \mathbb{R} \rightarrow \mathbb{R} $}\begin{itemize}
 \item[(i)] $f$ is non-negative, non increasing, and continuous on $\left[0, \infty \right) $.
 \item[(ii)] Let $\psi(s) = s~f(s)$ and $K= \sup \limits_{s\geq 0} \psi(s)$. The function $\psi$ is non decreasing and strictly increasing in the interval defined by $\psi<K$ with $p~<~K~<~\infty$.
 \item[(iii)] Let $P_N\left(\cdot\right)$ denotes the empirical distribution of $\mathbf{x}_1, \ldots, \mathbf{x}_N$. There exists $a > 0$ such that for all hyperplanes $W$ with $\text{dim}(W)\leq p-1$
\begin{align}
P_N\left(W\right) \leq 1 - \dfrac{p}{K} -a.
\end{align}
\end{itemize}
Because of non continuity of $f_i$ around $0$, the properties of M-estimators derived by Maronna cannot be applied directly to the estimators of the MGGD parameters. The objective of the next section is to derive similar properties for the estimator of $\mathbf{M}$ defined by the FP equation \eqref{eq:point_fixe_2}.
\end{rem}

\subsection{MGGD parameter estimation for unknown $\beta$}
When the shape parameter $\beta$ of the MGGD is unknown, the MLE of $\mathbf{M}$, $m$ and $\beta$ is obtained by differentiating the log-likelihood of $(\mathbf{x}_1, \ldots, \mathbf{x}_N)$ with respect to $\mathbf{M}$ and $m$ and $\beta$, i.e., by combining \eqref{MvM} and \eqref{Mvm} with the following relation
\begin{align}
\alpha(\beta) & = \dfrac{pN}{2\sum\limits_{i=1}^N y_i^{\beta}} \sum\limits_{i=1}^N  \left[ y_i^{\beta} \ln(y_i) \right] - \dfrac{pN}{2\beta} \left[ \Psi\left( \frac{p}{2\beta} \right) + \ln2 \right] \nonumber \\
\label{eq:MV_beta}
& - N - \dfrac{pN}{2\beta} \ln\left( \dfrac{\beta}{pN} \sum\limits_{i=1}^N y_i^{\beta} \right) = 0
\end{align}
where $\Psi(\cdot)$ is the digamma function. Equation~\eqref{eq:point_fixe_2} shows that $\mathbf{M}$ and $\beta$ can be estimated independently from the scale parameter $m$.


\section{Statements of the main results}\label{sec:main}
As the estimation scenario presented in the previous section has some similarities with the FP estimator studied in \cite{Pasc-08}, similar results about the estimator existence, uniqueness (up to a scale factor) and FP algorithm convergence are expected to be true. This section summarizes the properties of the FP estimator defined by \eqref{eq:point_fixe_2} for a known value of $\beta \in (0,1)$ (all proofs are provided in the appendices to simplify the reading). The case of an unknown value of $\beta$ will be discussed in the simulation section.

\subsection{Notations}
For any positive integer $n$, $\llbracket 1,n\rrbracket$ denotes the set of integers
$\{1,\hdots,n\}$. For any vector $\mathbf{x} \in \R^p$, $\|\mathbf{x}\|$ denotes the Euclidean norm of $\mathbf{x}$ such as $\|\mathbf{x}\|^2=\mathbf{x}^T \mathbf{x}$, where $\mathbf{x}^T$ is the transpose of $\mathbf{x}$. Throughout the paper, we will use several basic results about square matrices, especially regarding the diagonalization of real symmetric and orthogonal matrices. We invite the reader to consult \cite{hj} for details about these standard results. Denote as $M_p(\R)$ the set of $p \times p$ real matrices, $SO(p)$ the set of $p\times p$ orthogonal matrices and $\Mt$ the transpose of $\Gr{M}$. The identity matrix of $\Mm$ will be denoted as $\Gr{I}_p$. Several subsets of matrices used in the sequel are defined below
\begin{itemize}
\item [$\ast$] $\D$ is the subset of $\Mm$ defined by the symmetric positive definite matrices;
\item [$\ast$] $\overline{\D}$ is the closure of $\D$ in $\Mm$, i.e., the subset of $\Mm$ defined by the symmetric non negative definite matrices;
\item [$\ast$] For all $\alpha > 0$
 \begin{equation}
 \begin{array}{l} \D(\alpha) = \left\{\Gr{M} \in \D \; |\,\, ||\Gr{M}|| = \alpha\right\},  \\ \overline{\D}(\alpha)=\left\{\Gr{M} \in \Db \; |\,\, ||\Gr{M}|| = \alpha\right\}, \end{array}
 \end{equation}
 where $\Db(\alpha)$ is a compact subset of $\Mm$, $\vert\vert \cdot \vert\vert$ being the Frobenius norm.
\end{itemize}
For $\Gr{M}\in \D$, we introduce the open-half line spanned by $\Gr{M}$ defined by $\L_{\Gr{M}} = \{\lambda \,\Gr{M}, \; \lambda>0 \}$. Note that the order associated with the cone structure of $\D$ is called the Loewner order for symmetric matrices of $\Mm$ and is defined as follows: for any pair of two symmetric $p\times p$ real matrices  $(\Gr{A},\Gr{B})$,  $\Gr{A} \leq \Gr{B}$ ($\Gr{A}<\Gr{B}$ respectively) means that the quadratic form defined by $\Gr{B}-\Gr{A}$ is non negative (positive definite respectively), i.e., for all non zero $\Gr{x} \in\R^p$, $\Gr{x}^T  \, (\Gr{A} - \Gr{B}) \, \Gr{x} \geq 0$, ($>0$ respectively). Using that order, one has $\Gr{M} \in \D$ ($\in \Db$ respectively) if and  and only if $\Gr{M} > \Gr{0}$
($\Gr{M} \geq \Gr{0}$ respectively).

This section will make use of the following two applications
\begin{equation} \label{functionF}
\begin{array}{lllll}
F_{\chi}& : & \D & \rightarrow & \R^{+} \setminus \{0\} \\
& & \Gr{M} & \mapsto & |\Gr{M}|^{-1}\,\disp
\bep\sum_{i=1}^N y_i^{\beta}\enp^{-p/\beta}
\end{array}
\end{equation}
and
\begin{equation} \label{functionf}
\begin{array}{lllll}
f_{\chi}& : & \D & \rightarrow & \D \\
& & \Gr{M} & \mapsto & \dfrac{1}{N} \disp \sum\limits_{i=1}^N \dfrac{Np}{  y_i + y_i^{1-\beta} \disp \sum_{j\neq i} y_j^{\beta}} \,\mathbf{x}_i \mathbf{x}_i^T
\end{array}
\end{equation}
where $\chi=(\mathbf{x}_1, \ldots, \mathbf{x}_N) $, $y_i = \mathbf{x}_i^T {\mathbf{M}}^{-1} \mathbf{x}_i$ and $\beta \in (0,1)$. The function $F_{\chi}$ is the likelihood of $(\mathbf{x}_1, \ldots, \mathbf{x}_N)$ in which the parameter $m$ has been replaced by its estimator \eqref{Mvm}, up to a multiplicative constant and a power factor. Indeed:
\begin{align*}
\prod_{i=1}^N p(\mathbf{x}_i\vert \mathbf{M}, \hat{m}, \beta) & = \left[ \frac{\beta \Gamma\left(\frac{p}{2}\right)}{\pi^{\frac{p}{2}} \Gamma\left( \frac{p}{2\beta} \right) 2^{\frac{p}{2\beta}}} \right]^N \exp\left( -\frac{pN}{2\beta} \right) \\
& \times \left( \frac{pN}{\beta} \right)^{\frac{pN}{2\beta}} \Big[ F_{\chi}(\Gr{M}) \Big]^{N/2}.
\end{align*}

It is clear that $F_{\chi}$ is homogeneous of degree zero whereas $f_{\chi}$ is homogeneous of degree one, i.e., for all $\lambda>0$ and $\Gr{M}\in \D$, one has
$$
F_{\chi}(\lambda \, \Gr{M})=F_{\chi}(\Gr{M}),\ \ \ f_{\chi}(\lambda \Gr{M})=\lambda f_{\chi}(\Gr{M}).
$$
In order to understand the relationships between the two functions $F_{\chi}$ and $f_{\chi}$, we can compute the gradient of $F_{\chi}$ at $\Gr{M}\in \D$. Straightforward computations lead to
\begin{equation} \label{nablaF}
\nabla F_{\chi}(\Gr{M}) = F_{\chi}(\Gr{M}) \, \Gr{M}^{-1} \, \left[ f_{\chi}(\Gr{M}) - \Gr{M} \right]  \, \Gr{M}^{-1}.
\end{equation}
Clearly $\Gr{M}$ is an FP of $f_{\chi}$ if and only if $\Gr M$ is a critical point of the vector field defined by $\nabla F_{\chi}$ on $\D$, i.e., $\nabla F_{\chi}(\Gr{M})=0$.

\begin{rem}\label{rem:4}
There are some close links between the MGGDs and the SIRV distributions (that are both specific elliptical distributions). However, all MGGDs are not SIRV distributions and conversely. As a consequence, the FP equation \eqref{eq:point_fixe_2} associated with the MGGDs relies on the function $f_{\chi}$ which differs from the FP equation studied in \cite{Pasc-08} (which corresponds to the particular case $\beta=0$) and from that of \cite{chitour2008exact} which corresponds to SIRVs with random multipliers $\tau$. Similarly, the shape of the function $F_{\chi}$ differs significantly from the likelihoods studied in \cite{Pasc-08} and \cite{chitour2008exact} that are defined as products of integrals depending on the unknown texture distribution (see for \cite{Pasc-08} and \cite{chitour2008exact} for more details).
\end{rem}

In the sequel, we also use $f^{n}$ for $n\geq 1$ to denote the $n$-th iterate of $f$, i.e., $f^{n}:=f\circ ... \circ f$, where $f$ is repeated $n$ times. We also adopt the standard convention $f^{0}:=\textrm{Id}_{\D}$, where $\textrm{Id}_{\D}$ is the identity function defined in $\D$. To finish this section, we introduce an important assumption about the vectors $\Gr{x}_i$ for $1\leq i \leq N$
\begin{itemize}
\item $(H)$: For any set of $p$ indices belonging to $\llbracket 1,N\rrbracket$ and satisfying $i(1)<...<i(p)$, the vectors $\Gr{x}_{i(1)}, \hdots, \Gr{x}_{i(p)}$ are linearly independent.
\end{itemize}
This hypothesis is a key assumption for obtaining all our subsequent results. Hypothesis $(H)$ has the following trivial but fundamental consequence that we state as a remark
\begin{rem}\label{rem:1bis}
For all $n$ vectors $\Gr{x}_{i(1)}, \hdots,\Gr{x}_{i(n)}$ with $1\leq n\leq p$, $1\leq i\leq N$, the vector space generated by $\Gr{x}_{i(1)}, \hdots,\Gr{x}_{i(n)}$ has dimension $n$.
\end{rem}

\subsection{Contributions}
The contributions of this paper are summarized in the following theorems with proofs outlined in the next section.
\begin{Thm}\label{main-th}
For a given value of $\beta \in (0,1)$, there exists $\Mc_{FP}\in \D$ with unit norm such that, for all $\alpha>0$, $f_{\chi}$ admits a unique FP of norm $\alpha>0$ equal to $\alpha \, \Mc_{FP}$. Moreover, $F_{\chi}$ reaches its maximum in $\L_{\Mc_{FP}}$, the open half-line spanned by $\Mc_{FP}$.
\end{Thm}

Consequently, $\Mc_{FP}$ is the unique positive definite $p\times p$ matrix of norm
one satisfying
\begin{equation}\label{pf0}
\Mc_{FP}=p \disp \sum\limits_{i=1}^N \dfrac{\mathbf{x}_i \mathbf{x}_i^T}{  \hat{y}_i + \hat{y}_i^{1-\beta} \disp \sum_{j\neq i}\hat{y}_j^{\beta}}
\end{equation}
where $\hat{y}_i = \mathbf{x}_i^T \Mc_{FP}^{-1} \mathbf{x}_i$.

\begin{rem}\label{rem_interior}
Theorem \ref{main-th} relies on the fact that $F_{\chi}$ reaches its maximum in $\D$. In order to prove this result, the function $F_{\chi}$ is continuously extended by the zero function on the boundary of $\D$, except for the zero matrix. Since $F_{\chi}$ is positive and bounded in $\D$, we can conclude (see Appendix \ref{appendix1} for details).
\end{rem}

\noindent As a consequence of Theorem~\ref{main-th}, the following result can be obtained.
\begin{Thm}\label{th2}
Let $S$ be the discrete dynamical system defined on $\D$ by the recursion\begin{equation}\label{dis}
\Gr{M}_{k+1} = f_{\chi}(\Gr{M}_k).
\end{equation}
Then, for all initial conditions $\Gr{M}_0\in \D$, the resulting sequence $(\Gr{M}_k)_{k\geq 0}$ converges to an FP of $f_{\chi}$, i.e., to a point where $F_{\chi}$ reaches its maximum.
\end{Thm}
Theorem \ref{th2} can be used to characterize, numerically, the points where $F_{\chi}$ reaches its maximum and the value of that maximum. Note that the algorithm defined by (\ref{dis}) does not allow the norm of the FP to be controlled. Therefore, for practical convenience, a slightly modified algorithm can be used in which a $\Gr{M}$-normalization is applied at each iteration. This modified algorithm is proposed in the following corollary
\begin{Coro}\label{coro1}
The recursion
\begin{equation}\label{dis_norm}
\Gr{M}_{k+1}' = \cfrac{f_{\chi}(\Gr{M}_k')}{\textrm{Tr} \left[ f_{\chi}(\Gr{M}_k') \right]}
\end{equation}
initialized by
\begin{equation}\label{init}
\Gr{M}_{0}' = \cfrac{\Gr{M}_0}{\tr \left( \Gr{M}_0 \right)}
\end{equation} yields a sequence of matrices $\{\Gr{M}_0', \hdots, \Gr{M}_k'\}$ which converges to the FP $ \Mc_{FP}$ up to a scaling factor. Moreover, the matrices $\{\Gr{M}_0', \hdots, \Gr{M}_k'\}$ are related to $\{\Gr{M}_0, \hdots, \Gr{M}_k\}$  by
\begin{displaymath}
\Gr{M}_{i}' = \cfrac{\Gr{M}_i}{\tr(\Gr{M}_i)}, \quad 1 \leq i \leq k.
\end{displaymath}
\end{Coro}


\section{Proof outline}\label{section4}
\noindent This section provides the proofs of Theorems \ref{main-th} and \ref{th2}. Each proof is decomposed into a sequence of lemmas and propositions whose arguments are postponed in the appendices. For the proofs that can be directly obtained from those of \cite{Pasc-08}, we refer to \cite{Pasc-08}. In these cases, the differences due to the definitions of the function $f_{\chi}$ and the MGGD model for the observed vectors $\Gr x_i$, for $i=1,...,N$, imply only slight modifications.

\subsection{Proof of Theorem~\ref{main-th}} \label{4.1}
The proof of Theorem~\ref{main-th} is the consequence of several propositions whose statements are listed below. The first proposition shows the existence of an FP satisfying \eqref{eq:point_fixe_2}.
%

\begin{Prop} \label{pro:1}
The supremum of $F_{\chi}$ in $\D$ is finite and is reached at a point $\Mc_{FP}\in \D$ with $\|\Mc_{FP}\|=1$. Therefore, $f_{\chi}$
admits the open-half line $\L_{\Mc_{FP}}$ as fixed points.
\end{Prop}

\begin{proof}
See Appendix~\ref{appendix1}.
\end{proof}

\bigskip
\noindent It remains to show that there is no other FP of $f_{\chi}$ than those belonging to $\L_{\Mc_{FP}}$. For that
purpose, it is sufficient to show that all FPs of $f_{\chi}$ are collinear. However, Corollary V.1 of \cite{Pasc-08} indicates that
all FPs of $f_{\chi}$ are collinear if all the orbits of $f_{\chi}$ are bounded in $\D$. We recall here that the {\it orbit} of $f_{\chi}$ associated with $\Gr{M}\in\D$ is the trajectory of the dynamical system $S$ defined in (\ref{dis}) starting at $\Gr{M}$ (See \cite{guckenheimer1983nonlinear} for more details about orbits in dynamical systems). Moreover, according to \cite{Pasc-08}, when a function $f_{\chi}$ admits an FP, every orbit of $f_{\chi}$ is bounded if the following proposition is verified.

\begin{Prop}\label{pro:2}
The function $f_{\chi}$ verifies the following properties
\begin{itemize}
\item (P1) :  For all $\Gr{M},\Gr{Q} \in \D$, if $\Gr{M}\leq \Gr{Q}$, then $f_{\chi}(\Gr{M}) \leq f_{\chi}(\Gr{Q})$ (also true with strict inequalities);
\item (P2)~: for all $\Gr{M}, \Gr{Q} \in \D$, then
\begin{equation}
f_{\chi}(\Gr{M}+\Gr{Q}) \geq f_{\chi}(\Gr{M}) +f_{\chi}(\Gr{Q}),
\end{equation}
where equality occurs if and only if $\Gr{M}$ and $\Gr{Q}$ are collinear.
\end{itemize}
\end{Prop}

\begin{proof}
Since the function $f_{\chi}$ used in this paper differs from the one used in \cite{Pasc-08}, a specific analysis is required. It is the objective of Appendix~\ref{appendix2}.
\end{proof}

\bigskip

\noindent To summarize, Proposition \ref{pro:1} establishes the existence of matrices satisfying the FP equation \eqref{eq:point_fixe_2} while Proposition \ref{pro:2} together with the results of \cite{Pasc-08} can be used to show that there is a unique matrix of norm $1$ satisfying \eqref{eq:point_fixe_2}.


\subsection{Proof of Theorem~\ref{th2}}
In order to prove Theorem~\ref{th2}, we have to show that each orbit of $f_{\chi}$ converges to an FP of $f_{\chi}$. For that purpose, we consider for all $\Gr{M}\in \D$ the positive limit set $\omega(\Gr{M})$ associated with $\Gr{M}$, i.e., the set of cluster points of the sequence $(\Gr{M}_k)_{k \geq 0}$ when $k$ tends to infinity, where $\Gr{M}_{k+1} = f_{\chi}(\Gr{M}_k)$ and $\Gr{M}_0 = \Gr{M}$. Since the orbit of $f_{\chi}$ associated with $\Gr{M}$ is bounded in $\D$, the set $\omega(\Gr{M})$ is a compact of $\D$ and is invariant by $f_{\chi}$:  for all $\Gr{P} \in \omega(\Gr{M})$, $f_{\chi}(\Gr{P}) \in \omega(\Gr{M})$. It is clear that the sequence $(\Gr{M}_k)_{k \geq 0}$ converges if and only if $\omega(\Gr{M})$ reduces to a single point. According to \cite{Pasc-08}, $\omega(\Gr{M})$ reduces to a single point if the following proposition is satisfied.

\begin{Prop}\label{pro:3}
The function $f_{\chi}$ is eventually strictly increasing, i.e., for all $\Gr{Q},\Gr{P}\in\D$ such that $\Gr{Q}\geq \Gr{P}$ and $\Gr{Q}\neq \Gr{P}\,$, then
\begin{equation}
f_{\chi}^{p}(\Gr{Q}) > f_{\chi}^{p}(\Gr{P}).
\end{equation}
\end{Prop}

\begin{proof} Since the function $f_{\chi}$ used in this paper differs from the one used in \cite{Pasc-08}, a specific analysis is required. It is the objective of Appendix~\ref{appendix3}.
\end{proof}

\section{Simulations}
\label{sec:simultations}
This section presents simulation results to evaluate the performance of the MLE for the parameters of MGGDs. The first scenario considers i.i.d. $p$ dimensional data vectors $(\mathbf{x}_1, \ldots, \mathbf{x}_N)$ distributed according to an MGGD. These vectors have been generated using the stochastic representation~\eqref{eq:representation_stochastique} with a matrix $\mathbf{M}$ defined as
\begin{align}
\mathbf{M}(i,j) = \rho^{\vert i-j\vert} \; \textrm{for}  \; i,j \in \llbracket 0,p-1\rrbracket.
\end{align}
In the following, $1000$ Monte Carlo runs have been used in all experiments to evaluate the performance of the proposed estimation algorithms. Before analyzing the performance of the FP estimators based on \eqref{eq:point_fixe_2}, we illustrate the importance of the normalization $\tr \left( \mathbf{M} \right)=p$ advocated in this paper.

\subsection{Influence of the normalization}

The main advantage of the normalization (i.e., decomposition of $\mathbf{\Sigma}$ as the product $ m \times \mathbf{M}$, and trace constraint for the matrix $\mathbf{M}$) concerns the convergence speed of the algorithm. To illustrate this point, Fig.~\ref{fig:convergence_normalisation} shows the evolution of the criterion $D(k)$
	\begin{align}
	D(k) = \dfrac{ \vert\vert \hat{\mathbf{A}}_{k+1} - \hat{\mathbf{A}}_{k} \vert\vert }{ \vert\vert \hat{\mathbf{A}}_{k} \vert\vert },
	\end{align}
where $\hat{\mathbf{A}}_{k} = \hat{\mathbf{\Sigma}}_{k}$ for the blue curves and $\hat{\mathbf{A}}_{k} = \hat{m}_{k} \hat{\mathbf{M}}_{k}$ for the red curves. $\vert\vert \cdot \vert\vert$ is the Frobenius norm and $\hat{\mathbf{A}}_{k}$ is the estimator of $\mathbf{A}$ at step $k$. As observed, the convergence speed is significantly faster when a normalization condition is imposed at each iteration of the algorithm.

\begin{figure}[ht]
 \centering
	\includegraphics[width=\linewidth]{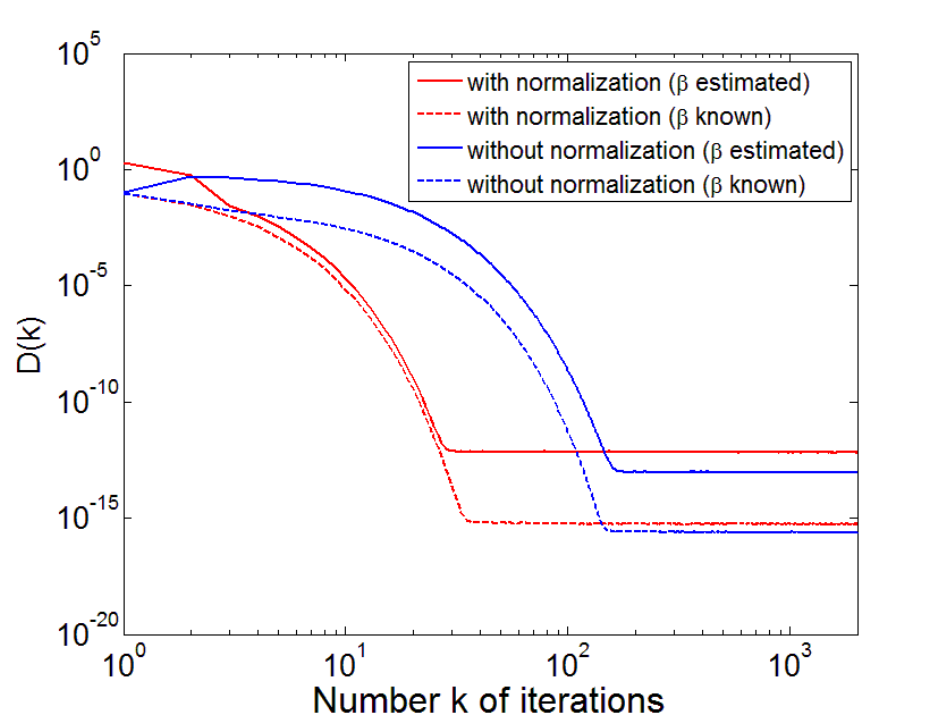}
 \caption{Variations of $D(k)$ versus number of iterations for $p=3$, $\beta=0.2$ and $\rho = 0.8$.}
 \label{fig:convergence_normalisation}
	\end{figure}

Fig.~\ref{fig:perf_normalisation_reviewer1_a} shows the evolution of the estimated bias and consistency of $\hat{\mathbf{A}}$ (the plain curves correspond to $\hat{\mathbf{A}} = \hat{m} \hat{\mathbf{M}}$ whereas $\hat{\mathbf{A}} = \hat{\mathbf{\Sigma}}$ for the dotted lines) versus the number of samples when $\beta$ is not estimated (the parameters are $\beta=0.2$, $\rho=0.8$ and $p=3$). The estimated bias of $\hat{\mathbf{A}}$ is defined as $\vert\vert \overline{\mathbf{A}} - \mathbf{A} \vert\vert$ where the operator $\overline{\mathbf{A}}$ is the empirical mean of the estimated matrices
\begin{align}
\label{eq:bias}
\overline{\mathbf{A}} = \dfrac{1}{I} \sum\limits_{i=1}^I \widehat{\mathbf{A}}(i).
\end{align}
For a given sample size, the experiment are performed $I$ times ($I=100$ in the following). Note that the bias criterion based on~\eqref{eq:bias} was used in~\cite{Pasc-08a} for assessing the performance of matrix estimators. Note also that other approaches based on computing the mean in the manifold of positive definite matrices could also be investigated \cite{sra2012new, fiori2009analgorithm,}. The estimated consistency of $\hat{\mathbf{A}}$ is verified by computing $\vert\vert \widehat{\mathbf{A}} - \mathbf{A} \vert\vert$. As observed, the estimation performance is the same when a normalization constraint for the scatter matrix is imposed or not.

\begin{figure}[ht]
 \centering
	\includegraphics[width=\linewidth]{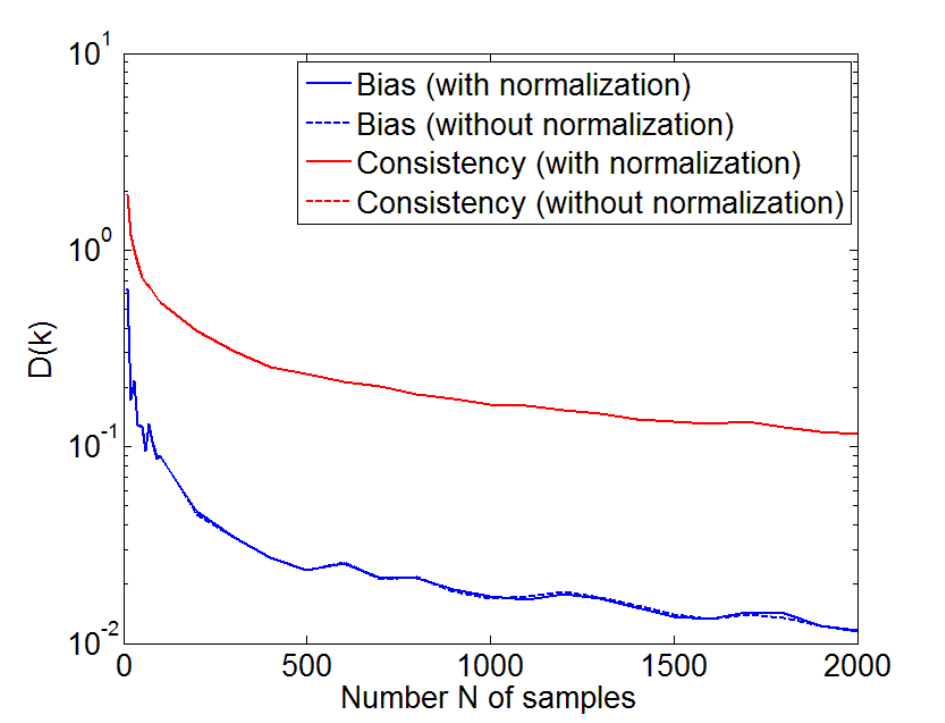}
 \caption{Influence of the normalisation of the scatter matrix on the estimation performance: estimated bias and consistency versus number of samples $N$.}
 \label{fig:perf_normalisation_reviewer1_a}
\end{figure}


A similar comment can be made for the shape parameter $\beta$ when this parameter is estimated (see Fig.~\ref{fig:perf_normalisation_reviewer1_b}). The Fisher information matrix has been recently derived for the parameters of MGGDs~\cite{Verd-11a}. It has been shown that this matrix only depends on the number $N$ of secondary data and the shape parameter $\beta$. The Cram\'er-Rao lower bounds (CRLBs) for the MGGD parameters can then be obtained by inverting the Fisher information matrix. These CRLBS provide a reference (in terms of variance or mean square error) for any unbiased estimator of the MGGD parameters. As observed in Fig.~\ref{fig:perf_normalisation_reviewer1_b}, the variance of $\hat{\beta}$ is very close to the Cram\'er-Rao lower bound for normalized or non-normalized scatter matrices.

\begin{figure}[ht]
 \centering
	\includegraphics[width=\linewidth]{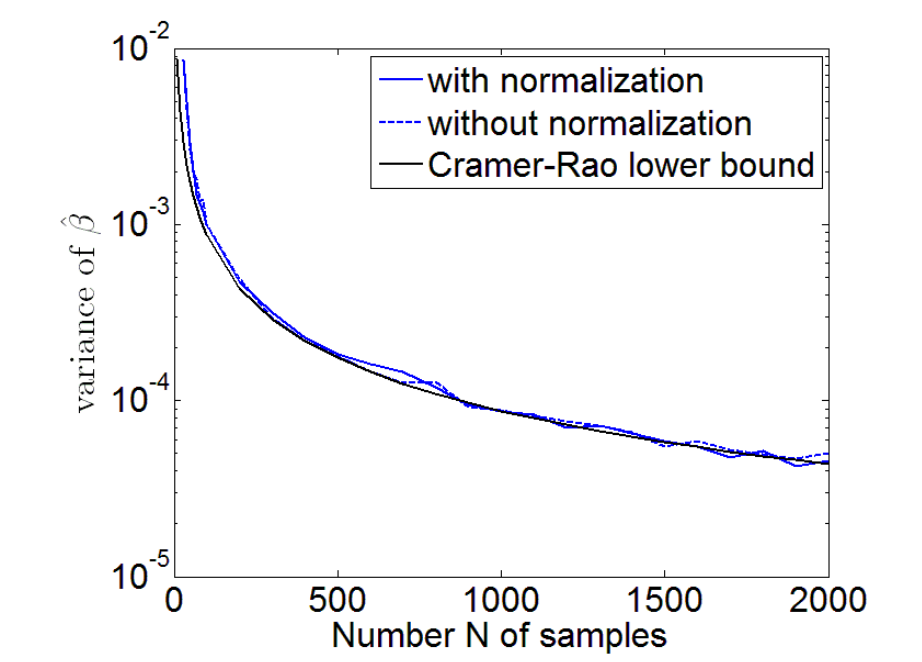}
 \caption{Variance of $\hat{\beta}$ versus number of samples $N$.}
 \label{fig:perf_normalisation_reviewer1_b}
\end{figure}

To summarize, the normalization of the scatter matrix (decomposition of $\mathbf{\Sigma}$ as the product $ m \times \mathbf{M}$, and trace constraint for the matrix $\mathbf{M}$) does not affect the statistical properties of the MLE. However, it ensures an increased convergence speed of the algorithm. Note also that a similar normalization was proposed in \cite[Eq. (15)]{Wiesel2012}.

\subsection{Known shape parameter}

\subsubsection{Convergence of the scatter matrix MLE}
Fig.~\ref{fig:res} shows some convergence results associated with the MLE of the scatter matrix $\mathbf{M}$. These results have been obtained for $p=3$, $\beta = 0.2$ (shape parameter) and $\rho = 0.8$. Convergence results are first analyzed by evaluating the sequence of criteria $C(k)$ defined as
\begin{align}
C(k) = \dfrac{ \vert\vert \widehat{\mathbf{M}}_{k+1} - \widehat{\mathbf{M}}_{k} \vert\vert }{ \vert\vert \widehat{\mathbf{M}}_{k} \vert\vert }
\end{align}
Fig.~\ref{fig:res}.(a) shows examples of criteria $C(k)$ obtained for various initial matrices $\mathbf{M}_0$ (``moments'' stands for $\mathbf{M}_0$ equal to the estimator of moments~\cite{Verd-11a}, ``identity'' stands for $\mathbf{M}_0 = \mathbf{I}_p$ and ``true'' corresponds to $\mathbf{M}_0 = \mathbf{M}$). After about $20$ iterations, all curves converge to the same values. Hence, the convergence speed of the proposed algorithm seems to be independent of its initialization. Fig.~\ref{fig:res}.(b) shows the evolution of criteria $C(k)$ for various numbers $N$ of secondary data. It can be observed that the convergence speed increases with $N$ as expected.

\begin{figure}[ht]
\begin{minipage}[b]{.8\linewidth}
  \centering
  \centerline{\includegraphics[width=\linewidth]{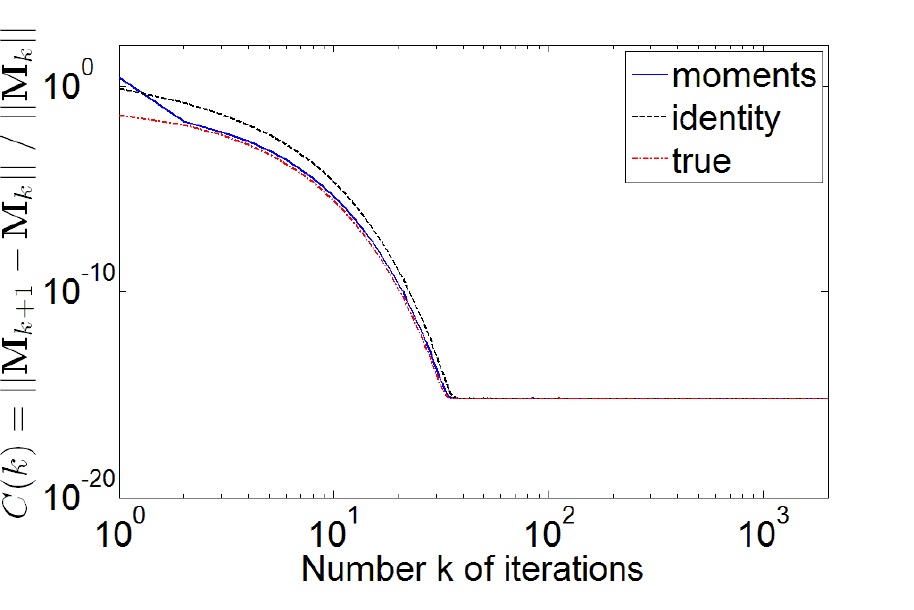}}
  \centerline{(a)}\medskip
\end{minipage}
\hfill
\begin{minipage}[b]{0.8\linewidth}
  \centering
  \centerline{\includegraphics[width=\linewidth]{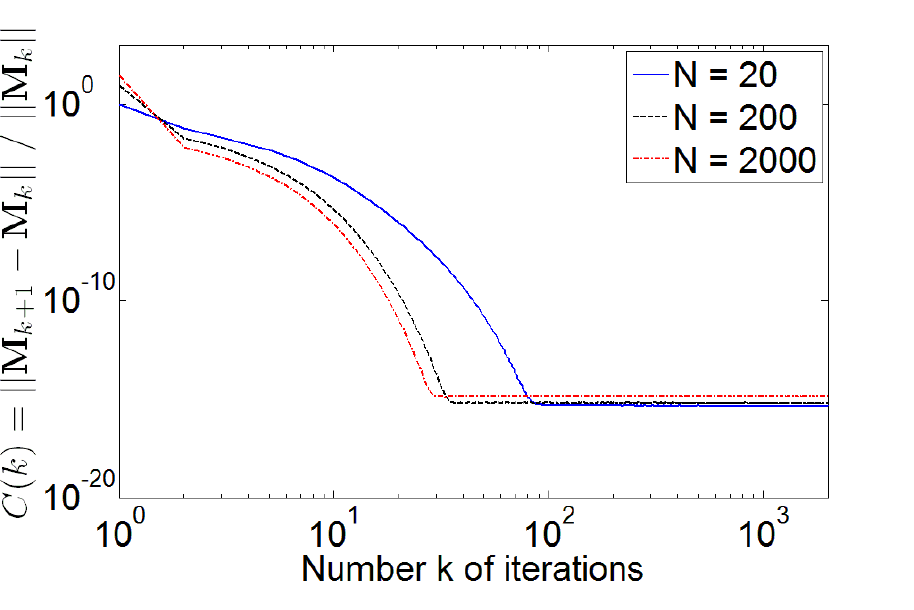}}
  \centerline{(b)}\medskip
\end{minipage}
\caption{Variations of $C(k)$ for $p=3$, $\beta=0.2$ and $\rho = 0.8$. (a) $C(k)$ versus number of iterations for different initializations ($N=200$). (b) $C(k)$ versus number of iterations for various values of $N$.}
\label{fig:res}
\end{figure}

\subsubsection{Bias and consistency analysis}

Fig.~\ref{fig:res_bias_consistency}.(a) shows the estimated bias of $\hat{\mathbf{M}}$ for different values of $\beta$ (precisely for $\beta \in \{0.2,0.5,0.8 \}$). As observed, the bias converges very fast to a small value which is independent of $\beta$.

Fig.~\ref{fig:res_bias_consistency}.(b) presents some results of consistency for the proposed estimator. Here, a plot of $\vert\vert \widehat{\mathbf{M}} - \mathbf{M} \vert\vert$ as a function of the number of samples $N$ is shown for different values of $\beta$ ($0.2$, $0.5$ and $0.8$). It can be noticed that this criterion tends to a	small value when $N$ increases for all values of $\beta$.

\begin{figure}[ht]

\begin{minipage}[b]{.8\linewidth}
  \centering
  \centerline{\includegraphics[width=\linewidth]{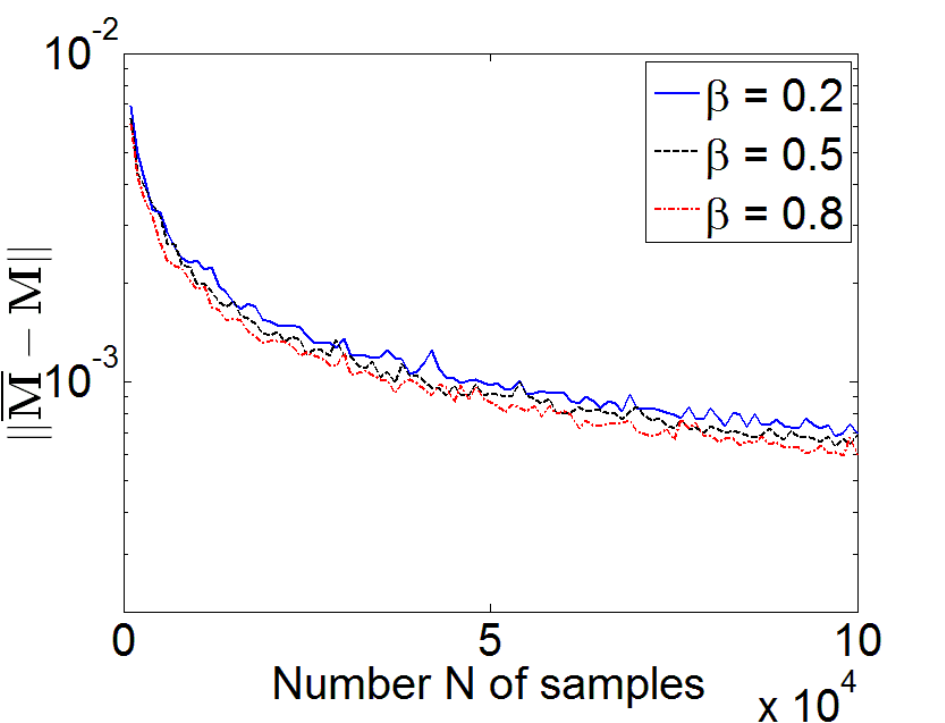}}
  
  \centerline{(a)}\medskip
\end{minipage}
\hfill
\begin{minipage}[b]{0.8\linewidth}
  \centering
  \centerline{\includegraphics[width=\linewidth]{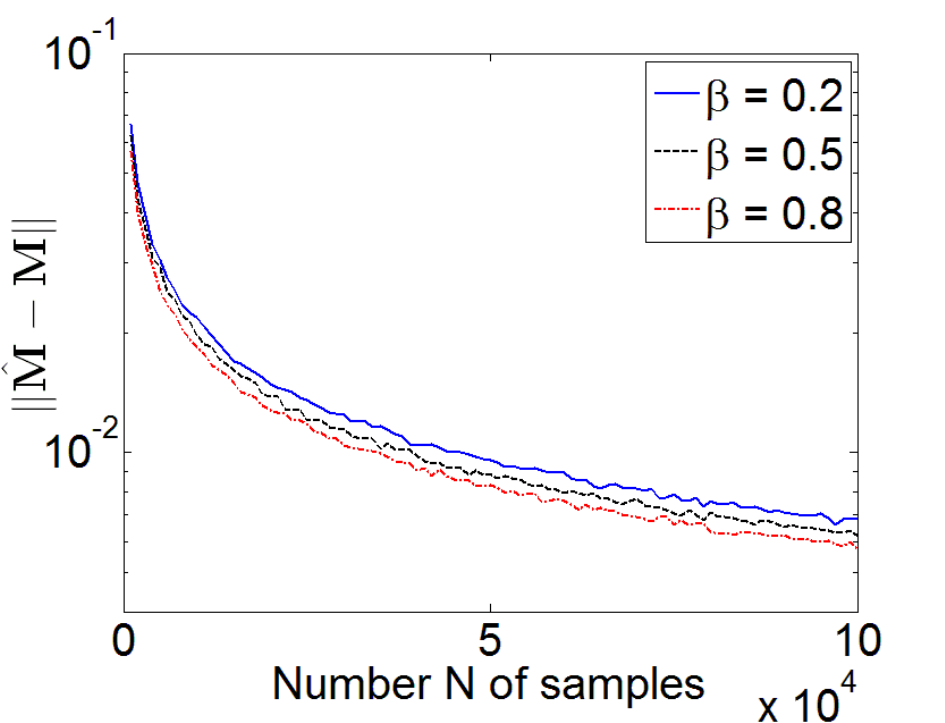}}
  
  \centerline{(b)}\medskip
\end{minipage}
\caption{(a) Estimated bias for different values of $\beta$, (b) estimated consistency for different values of $\beta$.}
\label{fig:res_bias_consistency}
\end{figure}

\subsection{Unknown shape parameter} \label{sec:unknown_shape}
When $\beta$ is unknown, the MLE of $\mathbf{M}$ and $\beta$ is defined by \eqref{eq:point_fixe_2} and \eqref{eq:MV_beta}. If $\mathbf{M}$ would be known, one might think of using a Newton-Raphson procedure to estimate $\beta$. The Newton-Raphson recursion based on \eqref{eq:MV_beta} is defined by the following recursion
\begin{align}
\label{eq:estimateur_beta_newton_raphson_sans}
\hat{\beta}_{n+1} = \hat{\beta}_n - \dfrac{\alpha(\hat{\beta}_n)}{\alpha^{'}(\hat{\beta}_n)}
\end{align}
where $\hat{\beta}_n$ is an estimator of $\beta$ at step $n$, and the function $\alpha(\beta)$ has been defined in~\eqref{eq:MV_beta}. In practice, when the parameters $\mathbf{M}$ and $\beta$ are unknown, we propose the following algorithm to estimate the MGGD parameters.
\begin{algorithm}
\caption{MLE for the parameters of MGGDs}
\label{algo1}
\begin{algorithmic}[1]
\STATE Initialization of $\beta$ and $\mathbf{M}$.
\FOR{$k = 1:\text{N}\_\text{iter}\_\text{max}$}
\STATE Estimation of $\mathbf{M}$ using one iteration of ~\eqref{eq:point_fixe_2} and normalization.
\STATE Estimation of $\beta$ by a Newton-Raphson iteration combining~\eqref{eq:MV_beta} and~\eqref{eq:estimateur_beta_newton_raphson_sans}.
\ENDFOR
\STATE Estimation of $m$ using~\eqref{Mvm}.
\end{algorithmic}
\end{algorithm}

\subsubsection{Bias and consistency analysis}
Fig.~\ref{fig:comparaison_beta_estime_connu} shows a comparison of the algorithm performance when the shape parameter $\beta$ is estimated (solid line) and when it is known (dashed line). As observed, the simulation results obtained with the proposed algorithm are very similar to those obtained for a fixed value of $\beta$.

\begin{figure}[ht]
  \centering
  
  \includegraphics[width=\linewidth]{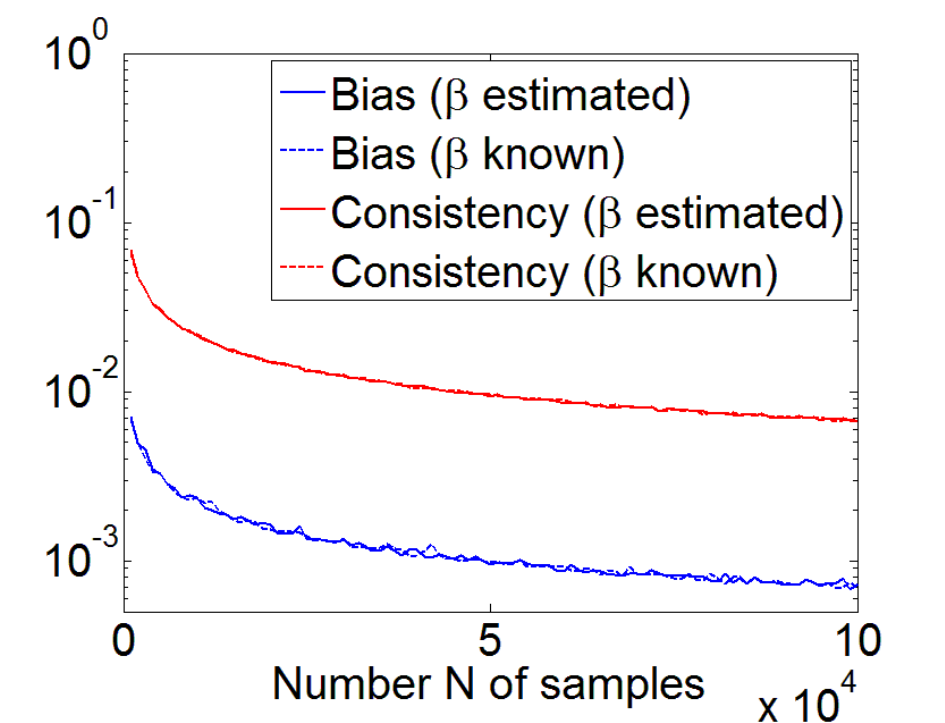}
  \caption{Estimated bias and consistency for $\beta=0.2$.}
  \label{fig:comparaison_beta_estime_connu}
\end{figure}


\subsubsection{Shape parameter $\beta$}
A comparison between the variances of estimators resulting from the method of moments and the ML principle as well as the correspondings CRLBs are depicted in Fig.~\ref{fig:beta} (versus the number of samples and the value of $\beta$). Fig.~\ref{fig:beta}.(a) was obtained for $\beta=0.2$, $\rho=0.8$ and $p=3$, while Fig.~\ref{fig:beta}.(b) corresponds to $N=10~000$, $\rho=0.8$ and $p=3$. The ML method yields lower estimation variances compared to the moment-based approach, as expected. Moreover, the CRLB of $\beta$ is very close to the variance of $\hat{\beta}$ in all cases, illustrating the MLE's efficiency.

\begin{figure}[ht]

\begin{minipage}[b]{.8\linewidth}
  \centering
  \centerline{\includegraphics[width=\linewidth]{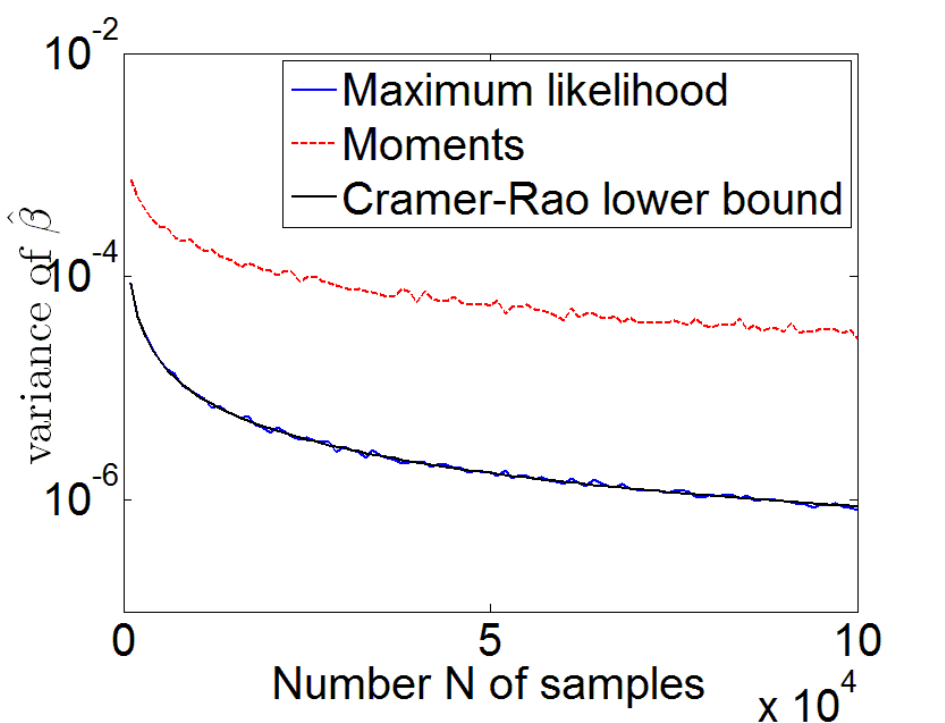}}
  
  \centerline{(a)}\medskip
\end{minipage}
\hfill
\begin{minipage}[b]{0.8\linewidth}
  \centering
  \centerline{\includegraphics[width=\linewidth]{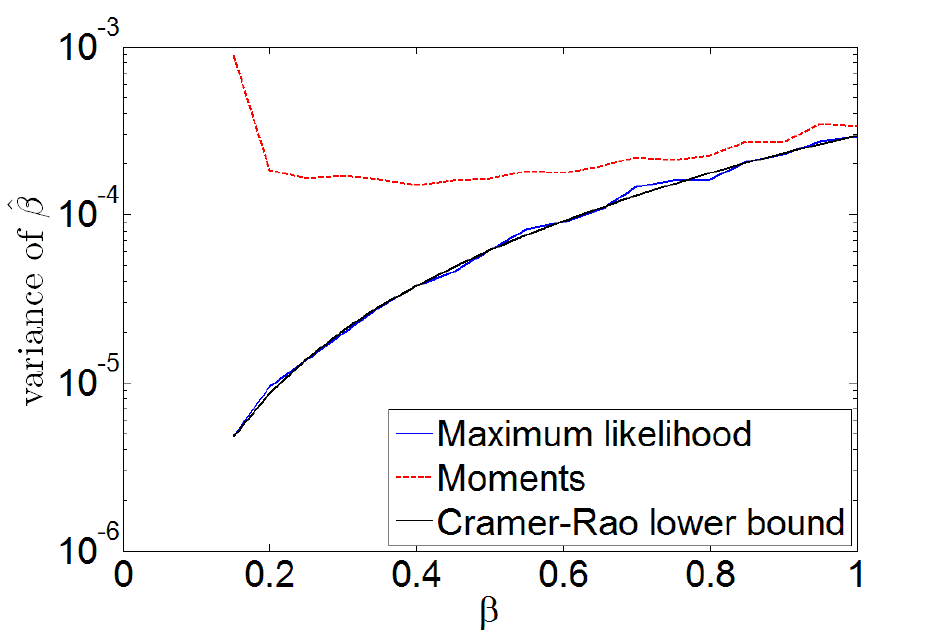}}
  \centerline{(b)}\medskip
\end{minipage}
\vspace{-0.6cm}
\caption{Estimation performance for parameter $\beta$. (a) Variance of $\hat{\beta}$ versus number of samples $N$ for $\beta=0.2$, $\rho=0.8$ and $p=3$, (b) Variance of $\hat{\beta}$ versus $\beta$ for $N=10~000$, $\rho=0.8$ and $p=3$.}
\label{fig:beta}
\vspace{-0.4cm}
\end{figure}

\subsection{Experiments in a real-world setting}
In this part, we propose to evaluate the performance of the MLE for the parameters of MGGDs encountered in a real-world application. MGGDs have been used successfully for modeling the wavelet statistics of texture images~\cite{Kwitt-04, Verd-11b}. In order to analyze the potential of MGGDs for texture modeling, we have considered two images from the VisTex database \cite{Vistex}, namely the ``Bark.0000'' and ``Leaves.0008'' images displayed in Fig.~\ref{fig:images_VisTex}.
The red, green and blue channels of these images have been filtered using the stationary wavelet transform with the Daubechies db4 wavelet. For the the first scale and orientation, the observed vector $\bf{x}$ (of size $p=3$) contains the realizations of the wavelet coefficients for each channel of the RGB image. MGGD parameters have then been estimated using the proposed MLE for an unknown shape parameter (Algorithm~\ref{algo1}), i.e., using the algorithm described in Section \ref{sec:unknown_shape}. The results are reported in Table~\ref{tab:MGGD_parameters}. Fig.~\ref{fig:histo_marginales_VisTex} compares the marginal distributions of the wavelet coefficients with the estimated MGGD and Gaussian distributions for the first subband of the red, green and blue channels (the top figures correspond to the image ``Bark.0000'' whereas the bottom figures are for the image ``Leaves.0008''). These results illustrate the potential of MGGDs for modeling color cue dependencies for texture images.

In the next experiments, we have generated $3$-dimensional data vectors $(\mathbf{x}_1, \ldots, \mathbf{x}_N)$ according to an MGGD with parameters given in Table~\ref{tab:MGGD_parameters}. Fig.~\ref{fig:beta_VisTex} shows the MLE performance for these parameters resulting from real texture images. As observed in Fig.~\ref{fig:beta_VisTex}, the performance of the MLE of $\mathbf{M}$ is very similar when $\beta$ is estimated or not (illustrating the unbiasedness and consistency properties of the scatter matrix estimator and the MLE efficiency of $\widehat{\beta}$ that have also been observed for synthetic data).

\begin{figure}[ht]

\begin{minipage}[b]{0.45\linewidth}
  \centering
  \centerline{\includegraphics[width=\linewidth]{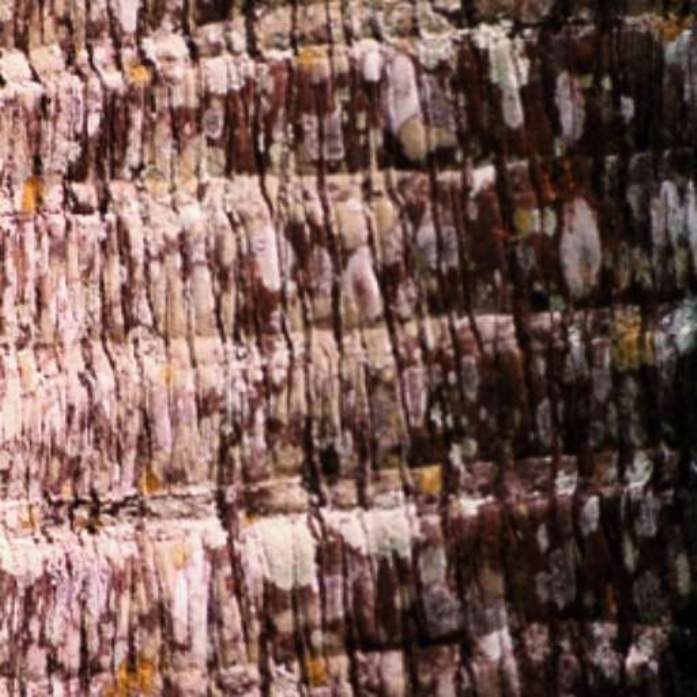}}
  \centerline{(a)}\medskip
\end{minipage}
\hfill
\begin{minipage}[b]{0.45\linewidth}
  \centering
  \centerline{\includegraphics[width=\linewidth]{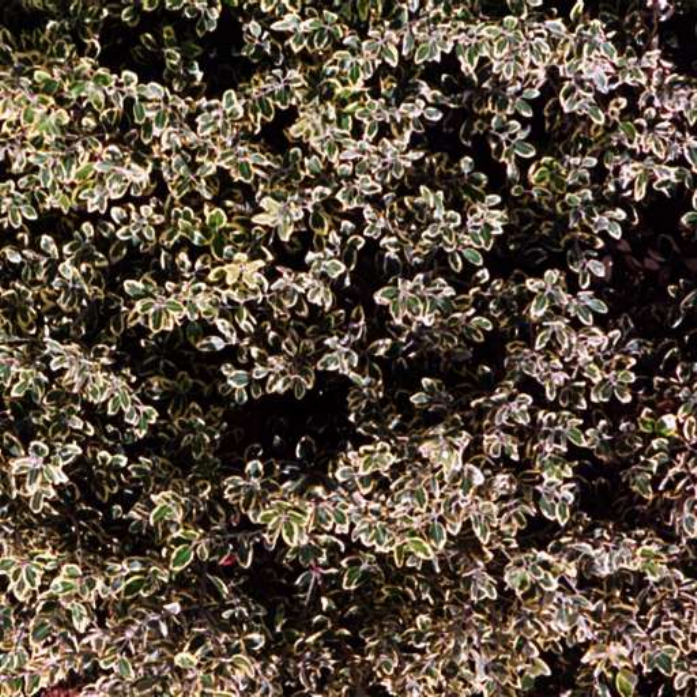}}
  \centerline{(b)}\medskip
\end{minipage}
\vspace{-0.6cm}
\caption{Images from the VisTex database. (a) Bark.0000 and (b) Leaves.0008.}
\label{fig:images_VisTex}
\vspace{-0.4cm}
\end{figure}

 \begin{table*}[!t]
    \caption{Estimated MGGD parameters for the first subband of the Bark.0000 and Leaves.0008 images.}
    \label{tab:MGGD_parameters}
    \begin{center}
       \begin{tabular}{||c|c|c|c||}
       \hline
       \hline
       Image & $\hat{m}$ & $\hat{\beta}$ & $\hat{\bf{M}}$  \\
       \hline
       Bark 0000 & 0.036 & 0.328 & $\begin{bmatrix}
0.988 & 0.992 & 0.883 \\
0.992 & 1.131 & 0.922 \\
0.883 & 0.922 & 0.881 \\
\end{bmatrix}$ \\
       \hline
       Leaves 0008 & 0.054 & 0.265 & $\begin{bmatrix}
0.935 & 0.966 & 0.871 \\
0.966 & 1.074 & 0.976 \\
0.871 & 0.976 & 0.991 \\
\end{bmatrix}$ \\
       \hline
       \hline
       \end{tabular}
    \end{center}
 \end{table*}

\begin{figure*}[ht]

\begin{minipage}[b]{0.3\linewidth}
  \centering
  \centerline{\includegraphics[width=\linewidth]{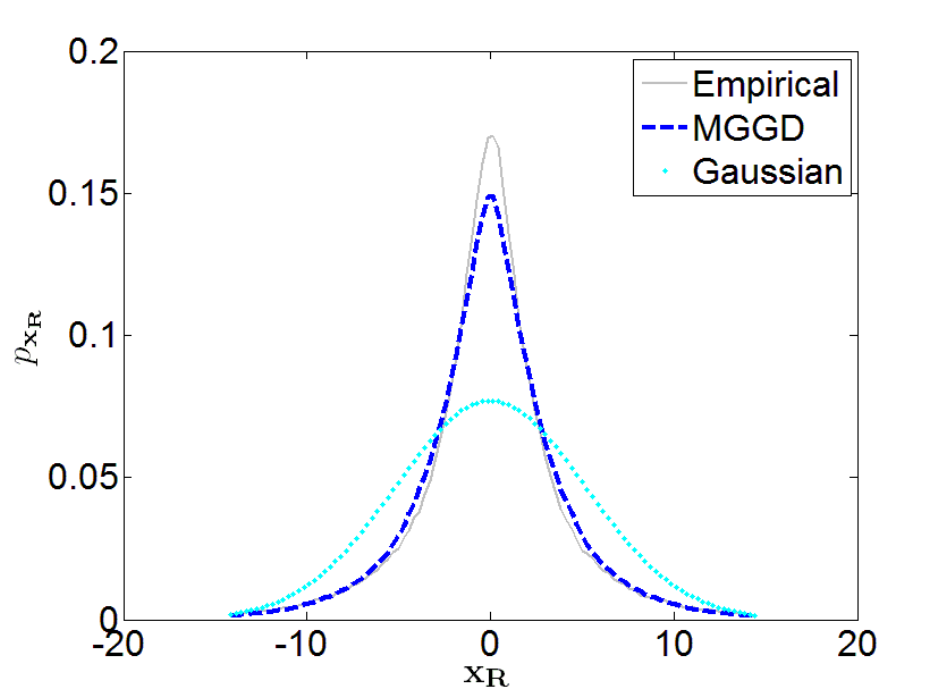}}
  \centerline{(a)}\medskip
\end{minipage}
\hfill
\begin{minipage}[b]{0.3\linewidth}
  \centering
  \centerline{\includegraphics[width=\linewidth]{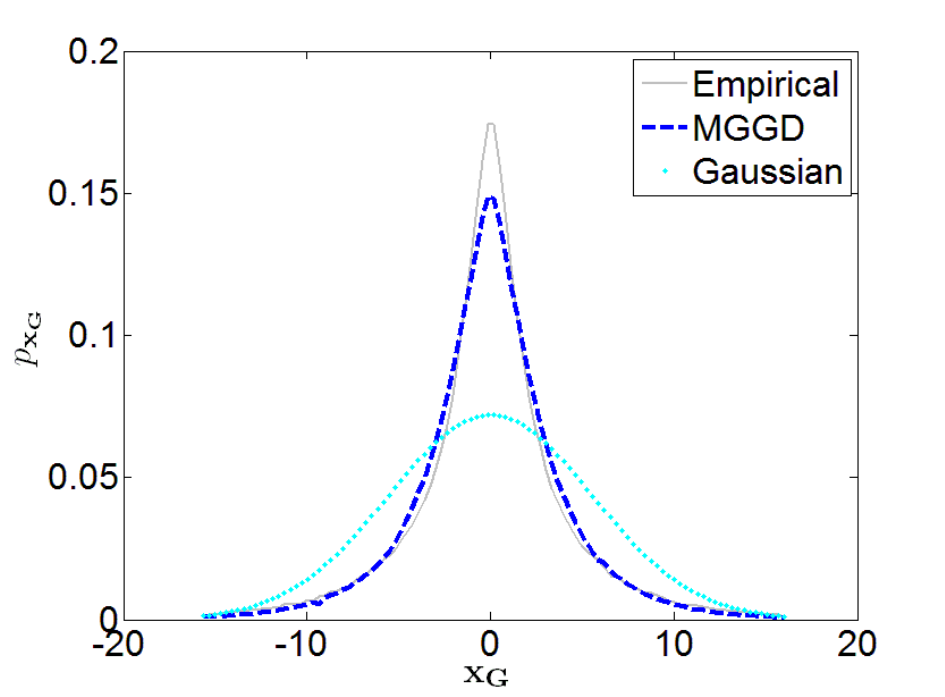}}
  \centerline{(b)}\medskip
\end{minipage}
\hfill
\begin{minipage}[b]{0.3\linewidth}
  \centering
  \centerline{\includegraphics[width=\linewidth]{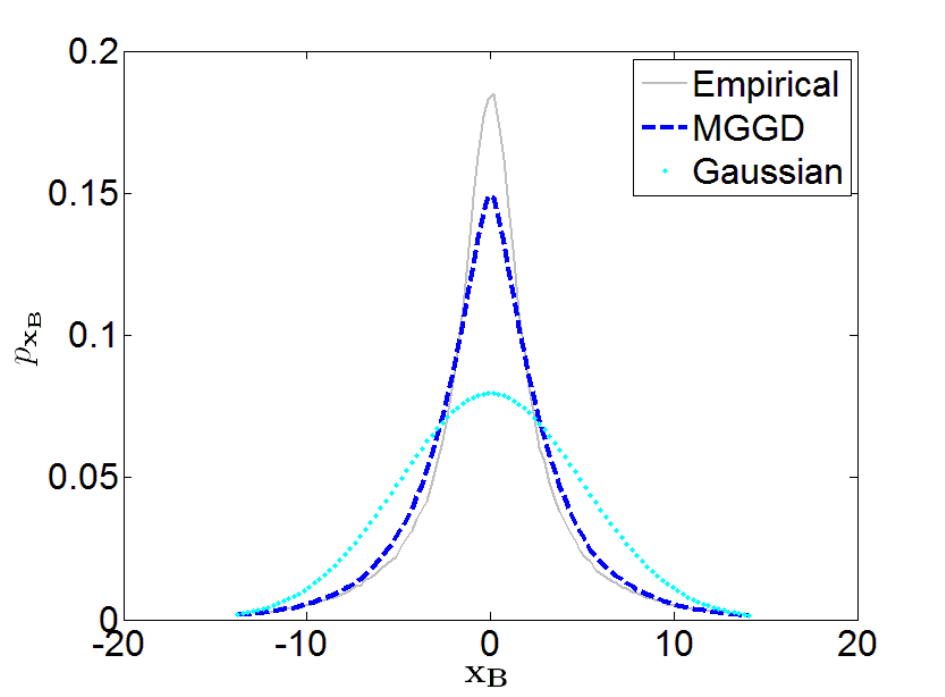}}
  \centerline{(c)}\medskip
\end{minipage}
\begin{minipage}[b]{0.3\linewidth}
  \centering
  \centerline{\includegraphics[width=\linewidth]{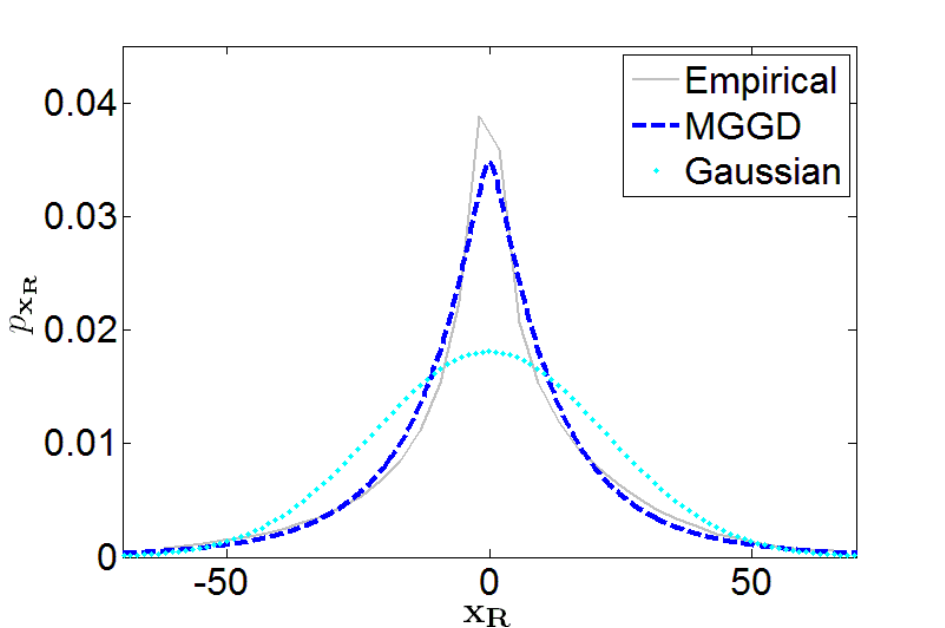}}
  \centerline{(d)}\medskip
\end{minipage}
\hfill
\begin{minipage}[b]{0.3\linewidth}
  \centering
  \centerline{\includegraphics[width=\linewidth]{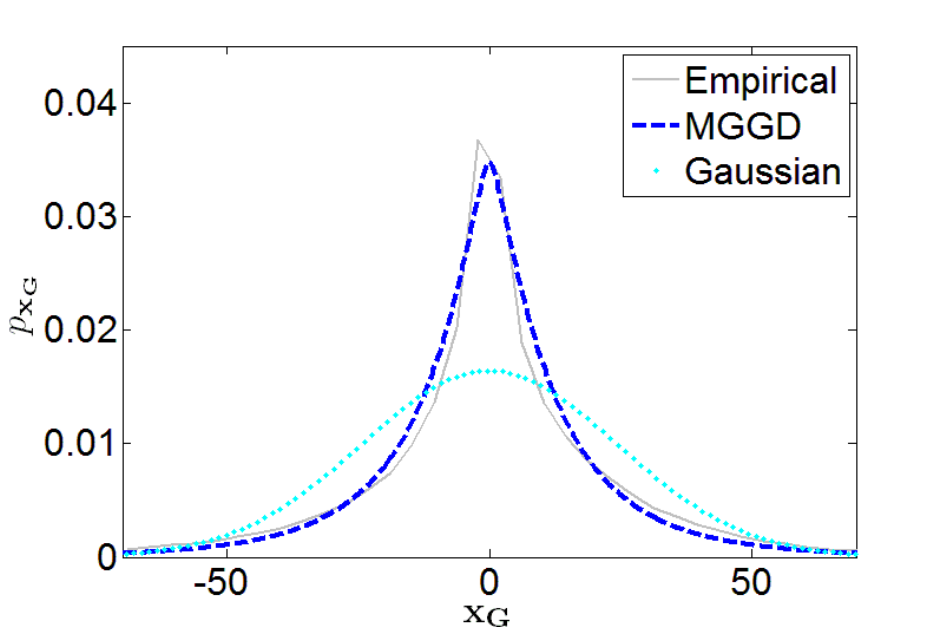}}
  \centerline{(e)}\medskip
\end{minipage}
\hfill
\begin{minipage}[b]{0.3\linewidth}
  \centering
  \centerline{\includegraphics[width=\linewidth]{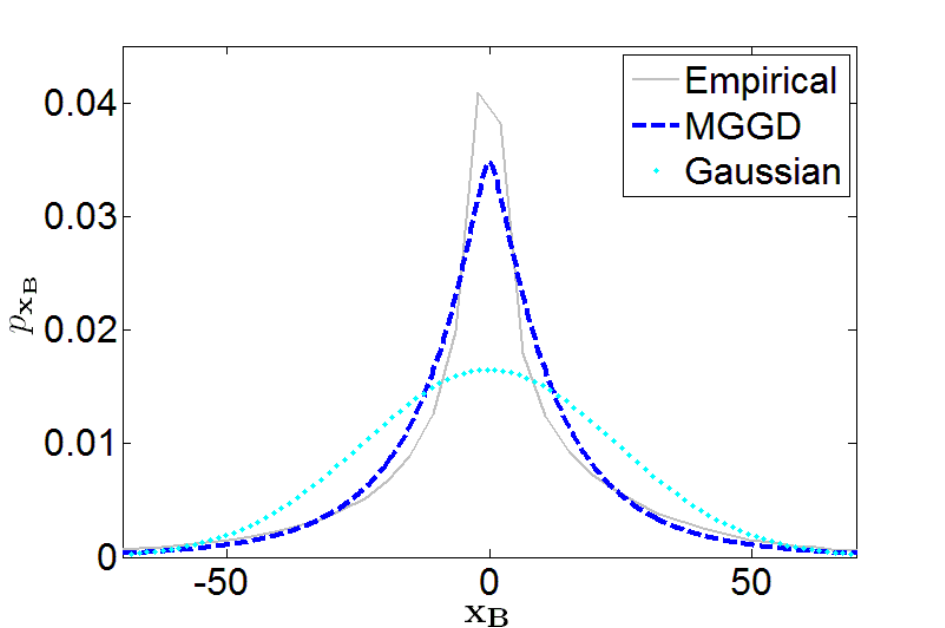}}
  \centerline{(f)}\medskip
\end{minipage}
\vspace{-0.6cm}
\caption{Marginal distributions of the wavelet coefficients with the estimated MGGD and Gaussian distributions of the first subband for the red, green and blue channels of the Bark.0000 (a,b,c) and Leaves.0008 images (d,e,f).}
\label{fig:histo_marginales_VisTex}
\vspace{-0.4cm}
\end{figure*}

\begin{figure*}[ht]

\begin{minipage}[b]{0.45\linewidth}
  \centering
  \centerline{\includegraphics[width=\linewidth]{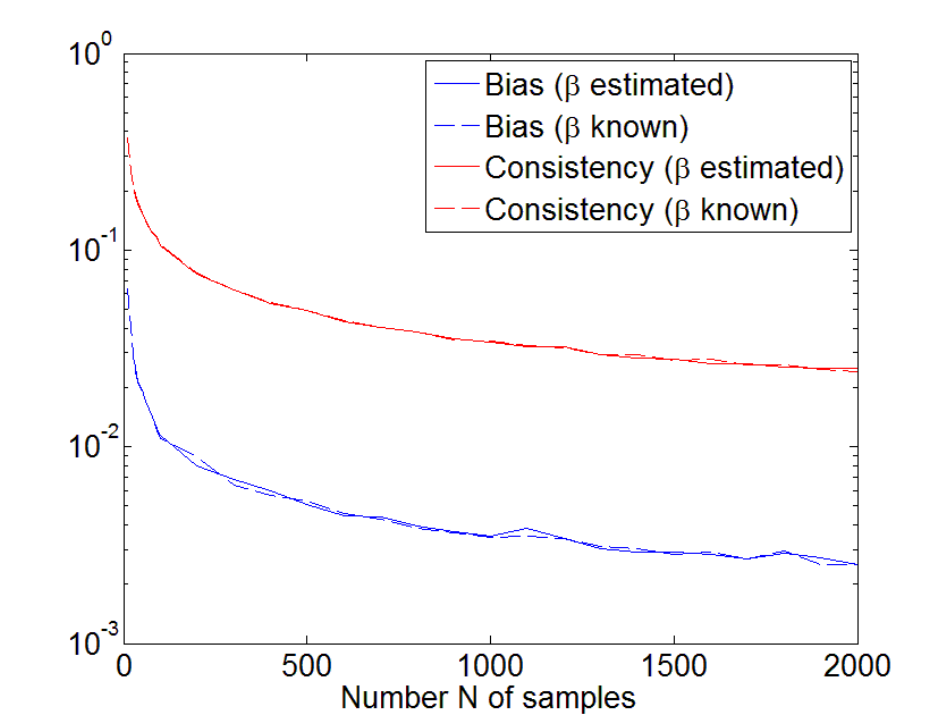}}
  \centerline{(a)}\medskip
\end{minipage}
\hfill
\begin{minipage}[b]{0.45\linewidth}
  \centering
  \centerline{\includegraphics[width=\linewidth]{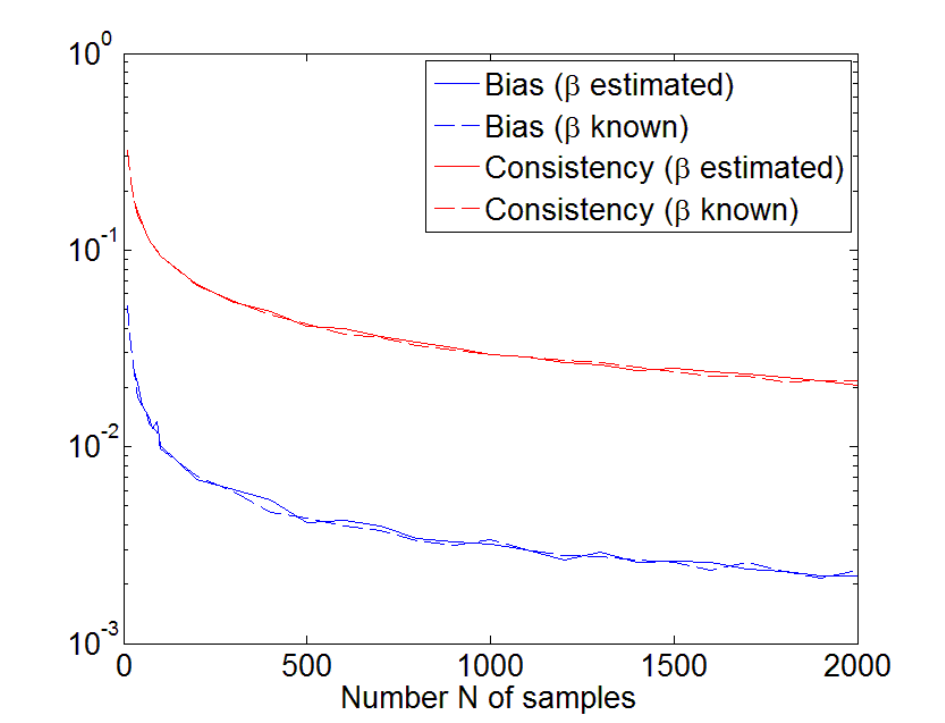}}
  \centerline{(b)}\medskip
\end{minipage}
\begin{minipage}[b]{0.45\linewidth}
  \centering
  \centerline{\includegraphics[width=\linewidth]{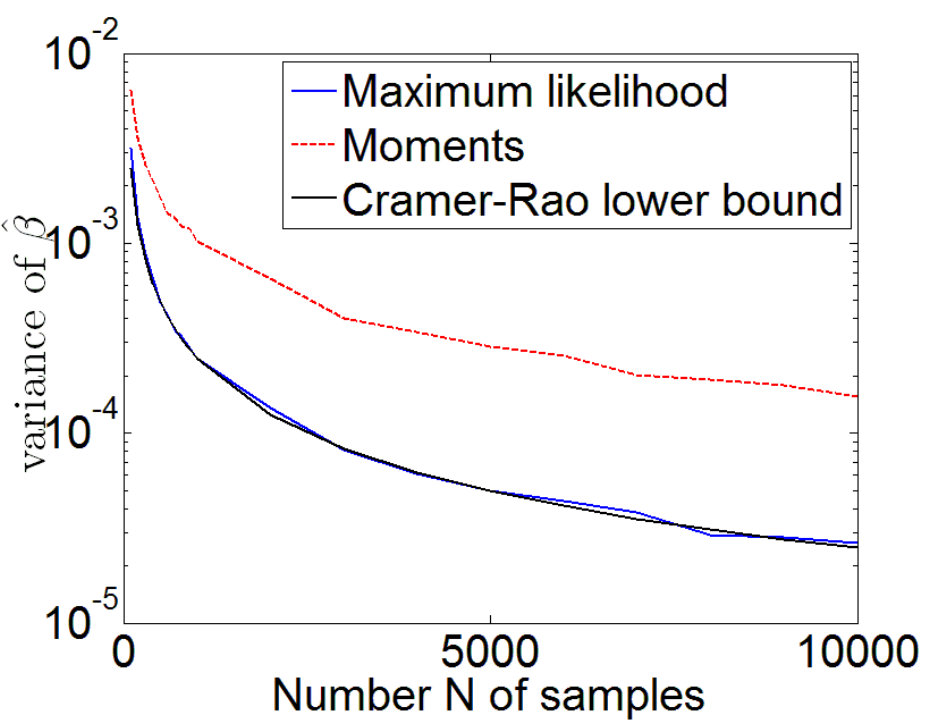}}
  \centerline{(c)}\medskip
\end{minipage}
\hfill
\begin{minipage}[b]{0.45\linewidth}
  \centering
  \centerline{\includegraphics[width=\linewidth]{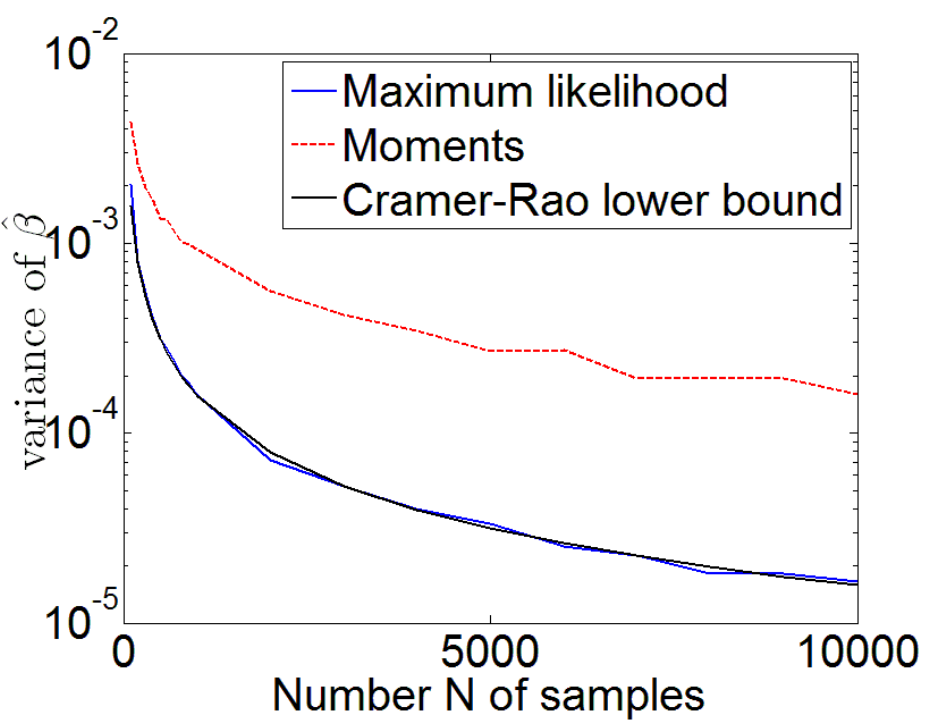}}
  \centerline{(d)}\medskip
\end{minipage}
\vspace{-0.6cm}
\caption{Estimation performance in a real-world setting. Estimated bias and consistency for (a) the Bark.0000 and (b) the Leaves.0008 settings. Variance of $\hat{\beta}$ versus number of samples $N$ for (c) the Bark.0000 and (d) the Leaves.0008 settings.}
\label{fig:beta_VisTex}
\vspace{-0.4cm}
\end{figure*}


\section{Conclusion}
\label{sec:conclusion}

This paper has addressed the problem of estimating the parameters of multivariate generalized Gaussian distributions using the maximum likelihood method. For any shape parameter $\beta \in (0, 1)$, we have proved that the maximum likelihood estimator of the scatter matrix exists and is unique up to a scalar factor. By setting to zero the partial derivative with respect to the scale parameter of the likelihood associated with generalized Gaussian distributions, we obtain a closed form expression of the scale parameter as a function of the scatter matrix. The profile likelihood is then obtained by replacing this expression in the likelihood. The existence of the maximum likelihood estimator of the scatter matrix was proved by showing that this profile likelihood is positive, bounded in the set of symmetric positive definite matrices and equals zero on the boundary of this set. We have also proved that for any initial symmetric positive definite matrix, the sequence of matrices satisfying a fixed point equation converges to the unique maximum of this profile likelihood. Simulations results have illustrated the unbiasedness and consistency properties of the maximum likelihood estimator of the scatter matrix. Surprisingly, these unbiasedness and consistency properties are preserved when the shape parameter of the generalized Gaussian distribution is estimated jointly with the other parameters. Further works include the use of multivariate generalized Gaussian distributions for various remote sensing applications including change detection, image retrieval and image classification.



\renewcommand{\theequation}{\thesection.\arabic{equation}}
\appendices


\section{Proof of Proposition~\ref{pro:1}}
\label{appendix1}
First, it is interesting to note that if $\Mc_{FP}$ is an FP of $f_{\chi}$, $\lambda \, \Mc_{FP}$ is also an FP of $f_{\chi}$ for all $\lambda>0$. This property is a direct consequence of the homogeneity of degree one of $f_{\chi}$. We start by demonstrating the following lemma.

\begin{Lemme}\label{le:1}
The function $F_{\chi}$ can be extended as a continuous function of $\Db\backslash \{\Gr{0}\}$ such that $F_{\chi}(\Gr{M}) = \Gr{0}$ for
all non invertible matrix $\Gr{M}\in \Db\backslash \{\Gr{0}\}$.
\end{Lemme}

\begin{proof}
It is enough to show that, for all non invertible $\Gr{M}\in \Db\backslash \{\Gr{0}\}$, and all sequence $(\Gr{Q}_k)_{k\geq 0}$ of $\overline{\D}$ converging to zero such that $\Gr{M} + \Gr{Q}_k$ is invertible, we have
\begin{equation}\label{eq2}
\lim_{k\to \infty} F_{\chi}(\Gr{M} + \Gr{Q}_k) = 0. \nonumber
\end{equation}

Using the definition of $F_{\chi}$ in \eqref{functionF}, the following result can be obtained for all $k\geq 0$
\setlength{\arraycolsep}{0.0em}
\begin{multline}
F_{\chi}(\Gr{M}+\Gr{Q}_k) = \\ \be \disp \sum_{i=1}^N \left[ \left|\Gr{M}+\Gr{Q}_k\right|\,\be\Gr{x}_{i}^T(\Gr{M}+\Gr{Q}_k)^{-1}\Gr{x}_{i}\en^p\right]^{\beta/p}\en^{-p/\beta}.\nonumber
\end{multline}
\setlength{\arraycolsep}{5pt}

\noindent Since $-p/\beta<0$, the conclusion holds true if $\exists i^*, 1\leq i^* \leq N$ such that
$$
\lim_{k\to\infty}\frac1{\left|\Gr{M}+\Gr{Q}_k\right|} \cfrac{1}{\left[\Gr{x}_{i^*}^T(\Gr{M}+\Gr{Q}_k)^{-1}\Gr{x}_{i^*}\right]^p}=0.
$$
which was demonstrated in \cite{Pasc-08} and concludes the proof.
\end{proof}

\bigskip

\underline{End of the proof of Proposition~\ref{pro:1}}\\
The end of the proof of Proposition~\ref{pro:1} is similar to the one given in \cite{Pasc-08}. Since $F_{\chi}$ is defined and continuous in the compact $\Db(1)$, this function reaches its maximum in $\Db(1)$ at a point denoted as $\Mc_{FP}$. Since $F_{\chi}$ is strictly positive in $\D(1)$ and equals 0 in $\Db(1) \backslash \D(1)$, the inequality $F_{\chi}(\Mc_{FP}) > 0$ leads to $\Mc_{FP} \in \D(1)$. In order to complete the proof of Proposition~\ref{pro:1}, we need to show the following lemma.

\begin{Lemme}\label{le:2}
Let $\Mc_{FP}\in \D(1)$ maximizing the function $F_{\chi}$. Then, $\nabla F_{\chi}(\Mc_{FP}) = \Gr{0}$, which implies that $\Mc_{FP}$ is an FP of $f_{\chi}$.
\end{Lemme}

\begin{proof}
Since the function $F_{\chi}$ defined in \eqref{functionF} differs from the one used in \cite{Pasc-08}, a specific analysis is required. By definition of $\Mc_{FP}$, one has
$$
F_{\chi}(\Mc_{FP}) = \disp \max_{\Gr{M} \in\D(1)} F_{\chi}(\Gr{M}).
$$
By defining $\NN(\Gr{M}) = \|\Gr{M}\|^2 -1$, one has $\NN(\Mc_{FP}) =0$. The Lagrange theorem ensures that $\nabla F_{\chi}(\Mc_{FP}) =
\lambda \nabla \NN(\Mc_{FP}) = 2 \lambda  \Mc_{FP}$ for $\lambda \ge 0$. Straightforward computations lead to
\begin{align*}
\nabla & F_{\chi}(\Gr{M}).\nabla\NN(\Gr{M}) =  \tr [\nabla F_{\chi}(\Gr{M})\,  \nabla \NN(\Gr{M}) ]  \\
& =  2F_{\chi}(\Gr M)\tr \left[ \Gr{M}^{-1}  \be f_{\chi}(\Gr{M}) - \Gr{M} \en \right] \\
& = 2F_{\chi}(\Gr M) \left(\tr \left[ \Gr{M}^{-1}  f_{\chi}(\Gr{M}) \right] -p \right)\\
& = 2F_{\chi}(\Gr M) \left(\dfrac{p}{N} \be \dfrac{1}{N}\disp \sum_{i=1}^N y_i^{\beta}\en^{-1} \disp \sum\limits_{i=1}^N \dfrac{\tr(\Gr M^{-1}\mathbf{x}_i \mathbf{x}_i^T)}{y_i^{1-\beta}} -p \right)\\
& = 2F_{\chi}(\Gr M) \left(\dfrac{p}{N} \be\dfrac{1}{N}\disp \sum_{i=1}^N y_i^{\beta}\en^{-1} \disp \sum\limits_{i=1}^N \dfrac{y_i}{y_i^{1-\beta}} -p \right) = 0.
\end{align*}
Since $\nabla F_{\chi}(\Mc_{FP})= 2 \lambda \Mc_{FP}$, one has $2\lambda=2\lambda~\|\Mc_{FP}\|^2 = \nabla F_{\chi}(\Mc_{FP}).\Mc_{FP} = 0$ which completes the proof of Lemma~\ref{le:2}.
\end{proof}


\section{Proof of Proposition~\ref{pro:2}}
\label{appendix2}
We start by establishing $(P1)$. Let $\Gr{M},\Gr{Q}\in\D$ with $\Gr{M} \leq \Gr{Q}$. Then, $\Gr{M}^{-1} \geq \Gr{Q}^{-1}$ and, for all $1 \leq i \leq N$, we have
\begin{multline*}
\cfrac{1}{ \Gr{x}_i^{\top} \, \Gr{M}^{-1} \, \Gr{x}_i + \be\Gr{x}_i^{\top} \, \Gr{M}^{-1} \, \Gr{x}_i\en^{1-\beta} \disp \sum_{j\neq i}\be\Gr{x}_j^{\top} \, \Gr{M}^{-1} \, \Gr{x}_j\en^{\beta}}
\leq \\ \cfrac{1}{ \Gr{x}_i^{\top} \, \Gr{Q}^{-1} \, \Gr{x}_i + \be\Gr{x}_i^{\top} \, \Gr{Q}^{-1} \, \Gr{x}_i\en^{1-\beta} \disp \sum_{j\neq i}\be\Gr{x}_j^{\top} \, \Gr{Q}^{-1} \, \Gr{x}_j\en^{\beta}},
\end{multline*}
which proves the property $(P1)$. The reasoning for the case with strict inequalities is identical.

We next turn to the proof of $(P2)$. Using the definition of $f_{\chi}$ in \eqref{functionf}, the following result can be easily obtained
\begin{equation}\label{fbis}
f_{\chi}(\Gr{M}) = \dfrac{p}{N} \be\dfrac{1}{N}\disp \sum_{i=1}^N y_i^{\beta}\en^{-1} \disp \sum\limits_{i=1}^N \dfrac{\mathbf{x}_i \mathbf{x}_i^T}{y_i^{1-\beta}}.
\end{equation}
For all unit vector $\Gr{x} \in \R^p$ such that  $\|\Gr{x}\| =1$ and all $\Gr{M} \in \D$, it is well known that
\begin{equation}\label{eq4}
\cfrac{1}{\Gr{x}^{\top} \, \Gr{M}^{-1} \, \Gr{x}} = \disp \min_{\Gr{z}^{\top} \, \Gr{x} \neq 0} \cfrac{\Gr{z}^{\top} \, \Gr{M} \, \Gr{z}}{(\Gr{x}^{\top}\, \Gr{z})^2},
\end{equation}
where the minimum is reached for the vectors $\Gr{z}$ belonging to the line generated by $\Gr{M}^{-1} \, \Gr{x}$. Moreover, since $\beta \in (0,1)$, one has
\begin{eqnarray}\nonumber
\be\cfrac{1}{\Gr{x}^{\top} \, \Gr{M}^{-1} \, \Gr{x}}\en^{1-\beta} & = & \be\disp \inf_{\Gr{z}^{\top} \, \Gr{x} \neq 0} \cfrac{\Gr{z}^{\top} \, \Gr{M} \, \Gr{z}}{(\Gr{x}^{\top}\, \Gr{z})^2}\en^{1-\beta}\\
\be\cfrac{1}{\Gr{x}^{\top} \, \Gr{M}^{-1} \, \Gr{x}}\en^{\beta} & = & \be\disp \inf_{\Gr{z}^{\top} \, \Gr{x} \neq 0} \cfrac{\Gr{z}^{\top} \, \Gr{M} \, \Gr{z}}{(\Gr{x}^{\top}\, \Gr{z})^2}\en^{\beta}.\nonumber
\end{eqnarray}
For $\Gr{M}, \Gr{Q} \in \D\,$, after noting that the function $f_{\chi}(\Gr{M})$ is unchanged if we replace each vector $\Gr{x_i}$ by the normalized vector $\Gr{n_i}=\Gr{x_i}/\| \Gr{x_i}\|$, the following results can be obtained
\begin{align*}
&f_{\chi}(\Gr{M}+\Gr{Q})  \\
& =  \dfrac{p}{N} g_\Gr{n}(\Gr{M}+\Gr{Q}) \disp \sum\limits_{i=1}^N \mathbf{n}_i \mathbf{n}_i^T \left[\min_{\Gr{z}^{\top} \, \Gr{n}_i
\neq 0} \cfrac{\Gr{z}^{\top} \, (\Gr{M}+\Gr{Q}) \, \Gr{z}}{(\Gr{n}_i^{\top}\, \Gr{z})^2} \right]^{1-\beta} \\
& = \dfrac{p}{N} g_\Gr{n}(\Gr{M}+\Gr{Q})\disp \sum\limits_{i=1}^N \mathbf{n}_i \mathbf{n}_i^T  \left( \min_{\Gr{z}^{\top} \Gr{n}_i
\neq 0} \left[ \cfrac{\Gr{z}^{\top} \, \Gr{M} \, \Gr{z}}{(\Gr{n}_i^{\top} \, \Gr{z})^2} +
\cfrac{\Gr{z}^{\top} \, \Gr{Q} \, \Gr{z}}{(\Gr{n}_i^{\top}\, \Gr{z})^2} \right] \right)^{1-\beta},
\end{align*}
where
$$
g_\Gr{n}(\Gr M) = \left[ \dfrac{1}{N}\disp \sum_{i=1}^N \Gr (\Gr{n}_i^T \Gr M^{-1} \Gr{n}_i)^{\beta} \right]^{-1}.
$$
More generally
$$
\disp \min_{\Gr{z} \in \A} \left[ f_1(\Gr{z} ) + f_2(\Gr{z} ) \right] \geq \disp \min_{\Gr{z} \in \A} f_1(\Gr{z} ) + \disp \min_{\Gr{z}  \in \A} f_2(\Gr{z} ),
$$
for all functions $f_1,f_2$ and set $\A$ giving a sense to the previous inequality. The same reasoning can be applied to the function $g_\Gr{n}(.)$ introduced above. Thus, $(P2)$ clearly holds true. It remains to study when equality occurs in $(P2)$. The property $(P2)$ becomes an equality if and only if, for all $1 \leq i \leq N$, one has $$\disp \min_{\Gr{z}^{\top} \, \Gr{n}_i \neq 0} \left[ \cfrac{\Gr{z}^{\top} \, \Gr{M} \, \Gr{z} }{(\Gr{n}_i^{\top} \, \Gr{z} )^2} + \cfrac{\Gr{z}^{\top} \,\Gr{Q} \, \Gr{z} }{(\Gr{n}_i^{\top}\, \Gr{z})^2} \right] =$$
\begin{equation}
\disp \min_{\Gr{z}^{\top} \, \Gr{n}_i \neq 0} \cfrac{\Gr{z}^{\top} \, \Gr{M} \, \Gr{z} }{(\Gr{n}_i^{\top} \, \Gr{z} )^2}\, + \disp \min_{\Gr{z}^{\top} \Gr{n}_i \neq 0} \cfrac{\Gr{z}^{\top} \, \Gr{Q} \, \Gr{z} }{(\Gr{n}_i^{\top}\, \Gr{z} )^2}, \label{eq5}
\end{equation}
which was shown in \cite{Pasc-08} to be true if and only if $\Gr{M}$ and $\Gr{Q}$ are collinear.\EOP


\section{Proof of Proposition~\ref{pro:3}}
\label{appendix3}
The proof of Proposition~\ref{pro:3} is similar to the proof of Proposition V.3 of \cite{Pasc-08}, even if the function $f_{\chi}$ used in this paper differs from the one defined in \cite{Pasc-08}. We first show that for all $\Gr{Q},\Gr{P}\in\D$, we have
\begin{equation}\label{inter}
\mbox{If }\Gr{Q} \geq \Gr{P}\mbox{ and }f_{\chi}(\Gr{Q}) = f_{\chi}(\Gr{P}), \mbox{ then }\Gr{Q} = \Gr{P}.
\end{equation}
Since $\Gr{Q} \geq \Gr{P}$ implies $\Gr{P}^{-1} -\Gr{Q}^{-1}\geq \Gr{0}$, for all $1 \leq i \leq N$, we have
$$
\cfrac{1}{\Gr{x}_i^{\top} \, \Gr{Q}^{-1}\,\Gr{x}_i} \geq \cfrac{1}{\Gr{x}_i^{\top} \,\Gr{P}^{-1}\, \Gr{x}_i}.
$$
Assuming $f_{\chi}(\Gr{Q}) = f_{\chi}(\Gr{P})$ and using hypothesis $(H)$ yields for all $1\leq i \leq N$,
\begin{align*}
(\Gr{x}_i^{\top} \,\Gr{Q}^{-1} \,\Gr{x}_i)^{1-\beta} \be\sum_{j=1}^N \Gr{x}_j^{\top} \,\Gr{Q}^{-1} \,\Gr{x}_j\en  = \\(\Gr{x}_i^{\top} \,\Gr{P}^{-1} \,\Gr{x}_i)^{1-\beta} \be\sum_{j=1}^N \Gr{x}_j^{\top} \,\Gr{P}^{-1} \,\Gr{x}_j\en.
\end{align*}
Moreover, assuming that there exists $i$ such that $\Gr{x}_i^{\top} \, \Gr{Q}^{-1}\,\Gr{x}_i \neq {\Gr{x}_i^{\top} \,\Gr{P}^{-1}\, \Gr{x}_i}$, i.e., $\Gr{x}_i^{\top} \, \Gr{Q}^{-1}\,\Gr{x}_i < {\Gr{x}_i^{\top} \,\Gr{P}^{-1}\, \Gr{x}_i}$ implies
\begin{align*}
(\Gr{x}_i^{\top} \,\Gr{Q}^{-1} \,\Gr{x}_i)^{1-\beta} \be\sum_{j=1}^N \Gr{x}_j^{\top} \,\Gr{Q}^{-1} \,\Gr{x}_j\en < \\ (\Gr{x}_i^{\top} \,\Gr{P}^{-1} \,\Gr{x}_i)^{1-\beta} \be\sum_{j=1}^N \Gr{x}_j^{\top} \,\Gr{P}^{-1} \,\Gr{x}_j\en
\end{align*}
which contradicts $f_{\chi}(\Gr{Q}) = f_{\chi}(\Gr{P})$. Thus, $f_{\chi}(\Gr{Q}) = f_{\chi}(\Gr{P})$ yields $\Gr{x}_i^{\top} \,\Gr{Q}^{-1} \,\Gr{x}_i = \Gr{x}_i^{\top} \,\Gr{P}^{-1} \,\Gr{x}_i$, for all $1 \leq i \leq N$. As a consequence, for all $1 \leq i \leq N$
$$
\Gr{x}_i^{\top}\,(\Gr{P}^{-1} -\Gr{Q}^{-1}) \, \Gr{x}_i = 0.
$$
Since $\Gr{P}^{-1} - \Gr{Q}^{-1}\geq \Gr{0}$, the previous equality indicates that $(\Gr{P}^{-1} - \Gr{Q}^{-1}) \,\Gr{x}_i = \Gr{0}$, for all $1 \leq i \leq N$. Using hypothesis $(H)$, the claim (\ref{inter}) is proved.\\

We now move to the proof of Proposition~\ref{pro:3}. We consider $\Gr{Q},\Gr{P}\in\D$ such that $\Gr{Q}\geq \Gr{P}$ and $\Gr{Q}\neq \Gr{P}$. As shown above, we have $f_{\chi}(\Gr{Q}) \geq f_{\chi}(\Gr{P})$ and $f_{\chi}(\Gr{Q}) \neq f_{\chi}(\Gr{P})$, which implies the existence of an index $i_0 \in \llbracket 1,N\rrbracket$ such that
$$
\xi_{i_0}:= \frac{p}{N} \be \cfrac{1}{\be\Gr{x}_{i_0}^{\top} \, \Gr{Q}^{-1}\, \Gr{x}_{i_0}\en^{1-\beta}} - \cfrac{1}{\be\Gr{x}_{i_0}^{\top} \, \Gr{P}^{-1} \, \Gr{x}_{i_0}\en^{1-\beta}} \en > 0.
$$
Note that for $\beta=0$, this result reduces to what was obtained in  \cite{Pasc-08}. Up to a relabel, we may assume that $i_0 = 1$, hence
\begin{equation}\label{eq7}
f_{\chi}(\Gr{Q}) \geq f_{\chi}(\Gr{P}) + \xi_1 \, \Gr{x}_1 \, \Gr{x}_1^{\top}
\end{equation}
which is the same result as the one obtained in \cite{Pasc-08}. As a consequence, the end of the proof of Proposition V.3 derived in \cite{Pasc-08} can be applied to our problem without any change.

\bibliographystyle{IEEEtran}
\bibliography{biblio2}

\end{document}